\documentclass[a4paper,USenglish,cleveref,thm-restate]{lipics-v2021}

\graphicspath{{./graphics/}}

\title{Optimizing Wiggle in Storylines}

\author{Alexander Dobler}{TU Wien, Vienna, Austria}{adobler@ac.tuwien.ac.at}{https://orcid.org/0000-0002-0712-9726}{Vienna Science and Technology Fund (WWTF)  grant [10.47379/ICT19035]}
\author{Tim Hegemann}{Universität Würzburg, Germany}{hegemann@informatik.uni-wuerzburg.de}{https://orcid.org/0009-0008-4770-3391}{Federal Ministry of Research, Technology and Space (BMFTR) grant [01IS22012C]}
\author{Martin Nöllenburg}{TU Wien, Vienna, Austria}{noellenburg@ac.tuwien.ac.at}{https://orcid.org/0000-0003-0454-3937}{Vienna Science and Technology Fund (WWTF)  grant [10.47379/ICT19035]}
\author{Alexander Wolff}{Universit\"at W\"urzburg, Germany \and \url{https://www.informatik.uni-wuerzburg.de/en/algo/team/wolff-alexander}}{}{https://orcid.org/0000-0001-5872-718X}{}
\authorrunning{A. Dobler, T. Hegemann, M. Nöllenburg, and A. Wolff} %

\Copyright{Alexander Dobler, Tim Hegemann, Martin Nöllenburg, and Alexander Wolff} %



\supplementdetails[subcategory={Source Code}]{Software}{https://github.com/hegetim/narrativiz}

\acknowledgements{We thank Marie Schmidt for discovering the similarities
between storylines and rolling stock schedules.  We thank Rowan
Hoogervorst and Boris Grimm for providing the rolling stock schedules
dataset and for their helpful feedback.}

\nolinenumbers %

\EventEditors{John Q. Open and Joan R. Access}
\EventNoEds{2}
\EventLongTitle{42nd Conference on Very Important Topics (CVIT 2016)}
\EventShortTitle{CVIT 2016}
\EventAcronym{CVIT}
\EventYear{2016}
\EventDate{December 24--27, 2016}
\EventLocation{Little Whinging, United Kingdom}
\EventLogo{}
\SeriesVolume{42}
\ArticleNo{23}

\hideLIPIcs

\usepackage[size=footnotesize]{todonotes}
\usepackage{xspace}
\usepackage{complexity}
\usepackage{rotating}
\usepackage{cite}
\usepackage{booktabs}
\usepackage{pgfplots}
\pgfplotsset{compat=1.18,small,scale only axis}

\DeclareMathOperator{\tm}{time}
\DeclareMathOperator{\charac}{ch}
\DeclareMathOperator{\Ac}{AC}
\DeclareMathOperator{\LWH}{LWH}
\DeclareMathOperator{\QWH}{QWH}
\DeclareMathOperator{\WC}{WC}
\newcommand{\di}{\ensuremath{\Delta}}
\newcommand{\da}{\ensuremath{\overline{\Delta}}}
\theoremstyle{definition}
\newtheorem{problem}{Problem}
\newcommand{\probname}[1]{{\normalfont\textsc{#1}}\xspace}
\newcommand{\WiggleHProb}{\probname{LWHMin}}
\newcommand{\WiggleCProb}{\probname{WCMin}}
\newcommand{\QdrWiggleHProb}{\probname{QWHMin}}
\newcommand{\Satprob}{\probname{Planar Monotone 3-Sat}}
\newcommand{\Satprobshort}{\probname{PM3-Sat}}

\usepackage{printlen}

\definecolor{dark blue}{rgb}{0.121,0.47,0.705}
\let\emph\relax\DeclareTextFontCommand{\emph}{\color{dark blue}\em}

\begin{document}
\maketitle

\begin{abstract}
A \emph{storyline visualization} shows interactions between characters over time.  
Each character is represented by an x-monotone curve.  
Time is mapped to the x-axis, and groups of characters that interact at a particular point $t$ in time must be ordered consecutively in the y-dimension at $x=t$. 
The predominant objective in storyline optimization so far has been the minimization of crossings between (blocks of) characters. 
Building on this work, we investigate another important, but less studied quality criterion, namely the minimization of \emph{wiggle}, i.e., the amount of vertical movement of the characters over time. 

Given a storyline instance together with an ordering of the characters at any point in time, we show that \emph{wiggle count minimization} is NP-complete.
In contrast, we provide algorithms based on mathematical programming to solve \emph{linear wiggle height minimization} and \emph{quadratic wiggle height minimization} efficiently.
Finally, we introduce a new method for routing character curves that focuses on keeping distances between neighboring curves constant as long as they run in parallel.

We have implemented our algorithms, and we conduct a case study 
that explores the differences between the three optimization objectives.
We use existing benchmark data, but we also present a new use case for storylines, namely the visualization of rolling stock schedules in railway operation.
\end{abstract}

\section{Introduction}

A \emph{storyline} can be seen as a temporal hypergraph; the vertices
represent \emph{characters} and the hyperedges, which correspond to
given points in time, represent \emph{meetings} (also called
\emph{interactions}) among the characters.  A \emph{storyline
  visualization} draws each character as an x-monotone curve and each
meeting as a vertical line segment at the x-coordinate that
corresponds to the point in time when the meeting happens; see
\cref{fig:teaser,fig:notationwiggle}a.  Storyline visualizations
have been made for books or movies
\cite{xkcd-storylines,gronemann2016,DoblerJJMMN-GD24,
  liu_storyflow_2013,tanahashi_design_2012} (where the meetings are
the scenes of the story), but they have also been used to illustrate
scientific collaboration \cite{hegemann-wolff-GD24} (where the
meetings are joint publications), for genealogical data
\cite{kch-tgdt-AVI10} (where the meetings are marriages and have
size~2), or for tweets following certain
topics~\cite{liu_storyflow_2013}.

\begin{figure}[tb]
\centering
\includegraphics[trim=0.5mm 0mm 0mm 5mm]{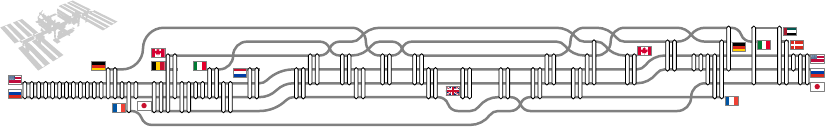}
\caption{A storyline that visualizes the nationalities of the crew members of expeditions to the International Space Station.
    Linear wiggle height is minimized using our LP formulation.}
\label{fig:teaser}
\end{figure}

In order to measure the quality of storyline visualizations, various
metrics have been suggested.  Most works have focused on reducing the
number of crossings of the character curves (simple pairwise
crossings~\cite{knpss-mcsv-GD15,gronemann2016} or so-called block
crossings~\cite{vandijk17, dlmw-csfbc-GD17}).  Others have also tried
to reduce the number of
wiggles~\cite{tanahashi_design_2012,liu_storyflow_2013,fn-mwsv-GD17}
(that is, the number of turns) and/or the amount of vertical
white-space \cite{lwwll-stes-TVCG13,tanahashi_design_2012}.

In this paper, we follow the storyline layout pipeline that has been
suggested by Liu, Wu, Wei, Liu, and
Liu~\cite[Fig.~2]{liu_storyflow_2013}.  It breaks down the overall
problem into four steps; 1.~hierarchy generation, 2.~ordering
(crossing minimization), 3.~alignment (wiggle minimization), and
4.~compaction (white-space reduction).  For step~1, Liu et al.\ assume
that the characters form a given hierarchy that must be respected in
step~2.  We assume that steps~1 and~2 have been settled, optimally or
heuristically.  Note that crossing minimization is
NP-hard~\cite{knpss-mcsv-GD15,vandijk17}.

This paper focuses on wiggle minimization, which appears in step~3
and, in a different form, in step~4.  We differentiate between three
variants; in each of them, we are given, for every point in time, the
vertical ordering of the characters, which for us is fixed.  In
\emph{wiggle count minimization} (\WiggleCProb) the task is to find, for each point
in time, y-coordinates for the character curves such that the total
number of inflection points (points where the curvature of the character curves changes sign) is
minimized. %
In \emph{linear wiggle height minimization} (\WiggleHProb), the total
change in y-coordinate is minimized (see the example in
\cref{fig:teaser}).  Finally, in \emph{quadratic wiggle height
  minimization} (\QdrWiggleHProb), the total sum of the squared wiggle heights is
minimized.

\subparagraph*{Related work.}
In terms of wiggle minimization in storylines, the existing literature mostly proposes heuristic methods, does not formally define a wiggle metric, or considers models that differ from our setting.
Ogawa and Ma~\cite{ogawa_software_2010} propose a greedy algorithm for storyline visualization that attempts to simultaneously minimize crossings and wiggle, however, wiggle is not formally defined. Tanahashi and Ma~\cite{tanahashi_design_2012} present a genetic algorithm for drawing storylines. They make use of \emph{slots}, which are horizontal strips of the visualization. A genome assigns each meeting to a slot. To evaluate a genome a pipeline approach is used, which rearranges lines inside slots, and then computes a fitness function measuring crossings and wiggle. Liu et al.~\cite{liu_storyflow_2013} also present a pipeline approach for drawing storylines, where characters further have geographic locations. Crossings are minimized first by applying a variant of the barycenter heuristic~\cite{DBLP:journals/tsmc/SugiyamaTT81} sequentially, then wiggle count is minimized heuristically by reducing the problem to the weighted longest common subsequence problem of two neighboring time steps. Lastly, quadratic wiggle height and whitespace is minimized using a quadratic program. The authors of \cite{tanahashi_efficient_2015} generalize the earlier approach of Tanahashi and Ma~\cite{tanahashi_design_2012} to streaming data, where time steps appear one by one. 
Arendt and Pirrung\cite{arendt_y_2017} consider a model where characters can split and merge. They first minimize crossings, and then reduce wiggle count minimization to the independent set problem, which they solve heuristically. %
Fröschl and Nöllenburg~\cite{froschl_minimizing_2018,fn-mwsv-18} proposed a model for storyline visualization where characters are assigned to discrete cells in a matrix for each time step. 
They used ILP and SAT solvers to simultaneously minimize a combination of crossings, wiggle count, and linear wiggle height.

Wiggle minimization also plays a role in \emph{streamgraphs} and \emph{stacked area charts} \cite{byronStackedGraphsGeometry2008,dibartolomeoThereMoreStreamgraphs2016,strungemathiesenAestheticsOrderingStacked2021}, which are visualizations of multiple time series vertically stacked on top of each other without gaps. %
In these charts, wiggle is minimized combinatorially by computing the best stacking order of the time series.

\subparagraph*{Our contribution}\hspace*{-2ex} is as follows.
\begin{itemize}
\item We present a linear program (LP) to solve \WiggleHProb efficiently; %
see
  \cref{section:linearprogram}.  
\item We give a quadratic program (QP) to solve \QdrWiggleHProb  efficiently; see \cref{sub:quadratic}.
\item We prove that \WiggleCProb is NP-hard, but can be formulated as
  an integer linear program (ILP) and admits an efficient solution for
  two consecutive time steps; see \cref{sec:wcmin}.
\item We present a new method for routing character curves that
  focuses on keeping distances between neighboring curves constant as
  long as they run in parallel; see \cref{sec:routing}.
\item We report the results of a case study with 17 benchmark instances
  in which we compare optimal solutions for the three objectives qualitatively and quantitatively; see
  \cref{sub:benchmark}.
\item We present the
  visualization of rolling stock schedules as a new use case for storylines; see \cref{sub:rolling}.
  This can help railway experts to improve or compare
  such schedules.
\end{itemize}

\section{Preliminaries}
\label{sec:preliminaries}

We use $[n]$ as shorthand for $\{1,2,\dots,n\}$.
A \emph{storyline instance} is a $4$-tuple
$(\mathcal{C}, T, \mathcal{M}, A)$, where
$\mathcal{C}=\{c_1,\dots, c_n\}$ is a set of \emph{characters},
$T=[\ell]$ %
is a set of totally ordered \emph{time steps}
(or \emph{layers}), and $\mathcal{M}=\{M_1,\dots,M_m\}$ is a set of
\emph{meetings}; see \cref{fig:notationwiggle}a.
\begin{figure}[tb]
  \centering
  \includegraphics{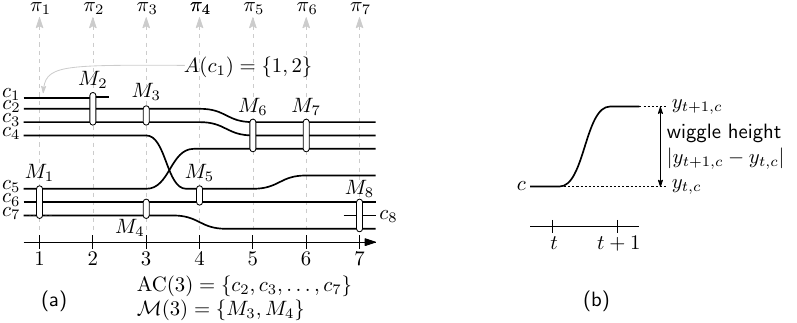}
  \caption{(a) Notation of storylines.
    (b) Wiggle height of character $c$ between two time steps.}
  \label{fig:notationwiggle}
\end{figure}
Each meeting $M\in \mathcal{M}$ has a
corresponding time step $\tm(M) \in [\ell]$ and consists of a set of
characters $\charac(M) \subseteq \mathcal{C}$.  Each character
$c\in \mathcal{C}$ is \emph{active} for a sequence
$A(c)=\{i, i+1, \dots, j\}$ of consecutive time steps.  For each
$t\in [\ell]$, we define the active character set
$\Ac(t)=\{c\in \mathcal{C}\mid t\in A(c)\}$ and the meeting set
 $\mathcal{M}(t)=\{M\in \mathcal{M}\mid \tm(M)=t\}$.

An \emph{ordered storyline instance} is furthermore given, for each $t\in [\ell]$, a permutation $\pi_t$ of the characters $\Ac(t)$. Here, for each meeting $M$ with $\tm(M)=t$, the characters $\charac(M)$ appear consecutively in $\pi_t$. 
We write $c\prec_t c'$ if $c$ comes before $c'$ in $\pi_t$.
A \emph{coordination} of a storyline instance defines for each $t$ and each $c\in \Ac(t)$ a y-coordinate $y_{t,c}\in \mathbb{R}$. The coordination is \emph{valid} w.r.t.\ an ordered storyline instance if for each $t\in [\ell]$ and each pair $c,c'\in \Ac(t)$ with $c\prec_{t}c'$, $y_{t,c}<y_{t,c'}$ holds. Given $\di,\da\in \mathbb{R}^+_0$, a coordination is \emph{$(\di,\da)$-nice} (or \emph{nice} if $\di$ and $\da$ are clear from the context) if it is valid and if, for every $t\in [\ell]$ and for every pair of characters $c,c'$ that are consecutive in $\pi_{t}$, it holds that
\begin{itemize}
    \item $|y_{t,c'}-y_{t,c}|=\di$ if $c$ and $c'$ are in the same meeting at time step $t$, and
    \item $|y_{t,c'}-y_{t,c}|\ge \da$ otherwise.
\end{itemize}
A coordination is \emph{integral} if its image contains only integers.

Given a coordination $y$ of an ordered storyline instance, its \emph{total (linear) wiggle height} (see~\cref{fig:notationwiggle}b) is defined as
\begin{equation}
    \LWH(y)=\sum_{t=1}^{\ell-1}\sum_{c\in \Ac(t)\cap \Ac(t+1)}|y_{t,c}-y_{t+1,c}|.
\end{equation}
Its \emph{total quadratic wiggle height} is defined as
\begin{equation}
    \QWH(y)=\sum_{t=1}^{\ell-1}\sum_{c\in \Ac(t)\cap \Ac(t+1)} (y_{t,c}-y_{t+1,c})^2.
\end{equation}
Finally, its \emph{total wiggle count} is defined as
\begin{equation}
    \WC(y)=\left|\{(t,c)\mid 1 \le t\le \ell-1, c\in \Ac(t)\cap \Ac(t+1), y_{t,c}\ne y_{t+1,c}\}\right|.
\end{equation}
Note that this matches our intuitive definition of wiggles as inflection points of the character curves.

\begin{problem}[\WiggleHProb]
Given an ordered storyline instance and $\di,\da\in \mathbb{R}$, find a nice coordination minimizing the total wiggle height.
\end{problem}

\begin{problem}[\QdrWiggleHProb]
Given an ordered storyline instance and $\di,\da\in \mathbb{R}$, find a nice coordination minimizing the total quadratic wiggle height.
\end{problem}

\begin{problem}[\WiggleCProb]
Given an ordered storyline instance and $\di,\da\in \mathbb{R}$, find a nice coordination minimizing the total wiggle count.
\end{problem}

\section{Wiggle Height Minimization}\label{sec:twh-lp}
In this section, we describe a linear program (LP)
for \WiggleHProb
and a quadratic program (QP) for \QdrWiggleHProb.
We also present observations about these programs.

\subsection{Linear Program for \WiggleHProb}\label{section:linearprogram}
Note that linear programs are polynomial-time solvable \cite{DBLP:journals/combinatorica/Karmarkar84}.
The linear program makes use of the following variables.
\begin{itemize}
    \item for $t\in [\ell],c\in \Ac(t)$, \emph{$y_{t,c}$} encodes the y-coordinate of character $c$ at time step $t$.
    \item for $t \in [\ell-1], c\in \Ac(t)\cap \Ac(t+1)$, \emph{$w_{t,c}$} encodes the wiggle height of character~$c$ between time steps~$t$ and $t+1$.
\end{itemize}
For our LP formulation we define, for each $t\in [\ell]$, the sets
\begin{align*}
N(t)&=\{(c,c')\mid c\text{ and }c'\text{ are characters that are consecutive in }\pi_t\text{ and }c\prec_t c'\},\\
N_{\mathcal{M}}(t)&=N(t)\cap \{(c,c')\mid c\text{ and }c'\text{ are in the same meeting at time step }t\},\text{ and }\\
N_A(t)&=N(t)\setminus N_{\mathcal{M}}(t)     
\end{align*}
The LP is given as follows.
\begin{align*}
    \text{Minimize } \sum_{t\in [\ell],c\in \Ac(t)} w_{t,c} \\
    \text{subject to } \hspace*{4ex}
    y_{t,c'}-y_{t,c}  & =\di          &  & \text{for all } t\in [\ell], (c,c')\in N_{\mathcal{M}}(t)       \tag{Cons-\di}\label{cons:di}                           \\
    y_{t,c'}-y_{t,c}  & \ge \da       &  & \text{for all } t\in [\ell], (c,c')\in N_A(t)       \tag{Cons-\da}\label{cons:da}                             \\
    y_{t,c}-y_{t+1,c} & \le w_{t,c} &  & \text{for all } t\in [\ell-1],c\in \Ac(t)\cap \Ac(t+1)\tag{W1}\label{cons:W1} \\
    y_{t+1,c}-y_{t,c} & \le w_{t,c}   &  & \text{for all } t \in [\ell-1],c\in \Ac(t)\cap \Ac(t+1)\tag{W2}\label{cons:W2} \\
    y_{t,c}           & \ge 0         &  & \text{for all }t\in [\ell],c\in\Ac(t)\tag{G0}\label{cons:G0}
\end{align*}
Constraints \eqref{cons:da} and \eqref{cons:di} ensure that the computed coordination is $(\di,\da)$-nice.
Together with the objective, \eqref{cons:W1} and \eqref{cons:W2} ensure that the wiggle value is computed correctly as $w_{t,c} = |y_{t,c}-y_{t+1,c}|$. 
The objective adds up the wiggle values. 
Our LP has the following nice property.

\begin{proposition}\label{thm:integerpolytope}
    If $\di,\da\in \mathbb{N}$, then all extreme points defined by the wiggle height minimization polytope are integer.
\end{proposition}

\begin{proof}
    Assume assignments to $y$- and $w$-variables corresponding to an extreme point.
    It is clear that $\min \{y_{t,c}\mid t\in [\ell],c\in \Ac(t)\}=0$, as otherwise there is some small $\epsilon>0$ such that each $y$-variable can be simultaneously increased or decreased by $\epsilon$, contradicting that we are at an extreme point. It is also clear that, for each $t \in [\ell-1]$ and $c\in \Ac(t)\cap \Ac(t+1)$, either \eqref{cons:W1} or \eqref{cons:W2} is binding.
    Now, let $G$ be the graph defined by $V(G)=\{y_{t,c}\mid t\in [\ell],c\in \Ac(t)\}$ and $E(G)=\{y_{t,c}y_{t',c'}\mid y_{t,c}-y_{t',c'}\in \mathbb{Z}\}$. If $G$ is connected, we are done as we have already seen that the smallest $y$-variable is integer. Otherwise, consider some connected component $C$ of $G$ such that $V(C)\not\subseteq \mathbb{Z}$.
    Note that it is possible to add to all $y_{t,c}\in V(C)$ some small $\epsilon>0$, resulting in a different point enclosed in the polytope. Equivalently, there exists $\epsilon'>0$ that can be subtracted from all $y_{t,c}\in V(C)$, witnessing in fact that we are not at an extreme point. Thus, $G$ is connected and the extreme point is integer.
\end{proof}

\subsection{Quadratic Program for \QdrWiggleHProb}\label{sub:quadratic}

The quadratic program makes use of the same variables and constraints as the LP above.
In fact, we just replace the objective function by the definition of quadratic wiggle height
\begin{align*}
    \text{Minimize } \quad  \sum_{t\in [\ell],c\in \Ac(t)} w_{t,c}^2.
\end{align*}

It is easy to see that, in matrix form, the objective 
has diagonal entries that are either~0 or $w_{t,c}$ (for some $t\in [\ell-1]$ and $c\in \Ac(t)$).
The entries of type~$w_{t,c}$ are non-negative by definition, hence the matrix of the objective is positive semidefinite.
QPs with such an objective and with linear constraints can be solved in polynomial time~\cite{bv-co-04}.

Note that we can replace each of the summands $w_{t,c}^2$ in the objective by
$(y_{t,c} - y_{t+1,c})^2$, which is the same as $(y_{t+1,c} - y_{t,c})^2$ (for all
$t \in [\ell-1], c \in \Ac(t) \cap \Ac(t+1)$).  Then we can drop the $w$-variables and 
the constraints \eqref{cons:W1} and \eqref{cons:W2}.

\section{Wiggle Count Minimization}
\label{sec:wcmin}

In this section we consider the number of wiggles in the visualization
as optimization criterion.
We show that for $\ell=2$ time steps, \WiggleCProb\ is polynomial-time solvable, while it is \NP-hard in the general case. We also present an integer linear program to solve \WiggleCProb.

\subsection{Polynomial Cases}
When ignoring the restriction that characters in meetings need to be equally spaced, finding a storyline visualization with the fewest number of wiggles is polynomial-time solvable. 
The problem is equivalent to finding the longest common subsequence between each pair of adjacent orderings $\pi_t,\pi_{t+1}$, $t\in [\ell-1]$.
Integer coordinates can be obtained by scaling all y-coordinates accordingly.
A similar observation has been made in \cite{liu_storyflow_2013}.
\begin{observation}
    Given an ordered storyline instance, one can compute in polynomial time a valid integral coordination that minimizes the total wiggle count.
\end{observation}
Rather surprisingly, we obtain the following result.  It shows that we
can compute the ``maximum wiggle-free'' storyline---a set of
characters that all can be realized without a single wiggle---in
polynomial time.
\begin{restatable}{theorem}{maxbendless}\label{theorem:maxbendless}
  Let $(\mathcal{C}, T, \mathcal{M}, A, (\pi_t)_{t \in [\ell]})$ be an
  ordered storyline instance.  Let~$\mathcal{C}'$ be the set of
  characters in~$\mathcal{C}$ that are active at all time steps, i.e.,
  $A(c)=[\ell]$ for all $c\in \mathcal{C}'$.  Then, given a pair
  $(\di,\da)$, a $(\di,\da)$-nice coordination with the largest
  wiggle-free subset of~$\mathcal{C}'$ can be computed in
  $\mathcal{O}(|\mathcal{C}|^2\cdot \ell)$ time.
\end{restatable}
\begin{proof}
    The idea is to find pairs of characters that can be realized simultaneously without wiggle.
    Clearly, such pairs must have the same relative orders in all $\pi_i$, $i\in [\ell]$.
    Thus, define $\mathcal{C}_<^2=\{(c,c')\in \mathcal{C}'^2\mid \forall i\in[\ell]: c\prec_i c'\}$ exactly as these pairs.
    Now, for each $(c,c')\in \mathcal{C}_<^2$ and $i\in [\ell]$, we compute
    $\operatorname{miny}_i(c,c')$ which is the minimum required space in y-direction for the sequence of characters between $c$ and $c'$ in $\pi_i$, including $c$ and $c'$. The value $\operatorname{maxy}_i(c,c')$ is defined analogously as the maximum required space.
    To compute these values, we define, for $i\in \ell$ and $1\le j_1<j_2\le |\Ac(i)|$, the following auxiliary table:
    \begin{align}
      C^{\mathrm{cons}}_i(j_1,j_2) &= |\{j_1\le k< j_2\mid \exists
      M\in \mathcal{M}(i):\{\pi_i(k),\pi_i(k+1)\}\subseteq \charac(M)\}|.
      \intertext{Now $\operatorname{miny}_i(c,c')$ and $\operatorname{maxy}(c,c')$
      can be computed as follows:}
      \operatorname{miny}_i(c,c') &= C^{\mathrm{cons}}_i(\pi_i^{-1}(c),\pi_i^{-1}(c'))\cdot \di+(\pi^{-1}_i(c')-\pi^{-1}(c)-C^{\mathrm{cons}}_i(\pi_i^{-1}(c),\pi_i^{-1}(c'))\cdot \da,\\
      \operatorname{maxy}_i(c,c') &=
      \begin{cases}
        \di\cdot (\pi_i^{-1}(c')-\pi_i^{-1}(c)) & \text{ if }\exists M\in
        \mathcal{M}(i):\{c, c'\}\subseteq \charac(M) \\
        \infty                                 & \text{ otherwise.}
      \end{cases}
      \intertext{The values $C^\mathrm{cons}_i$ can be computed in time
      $\mathcal{O}(|\mathcal{C}|^2\cdot \ell)$ by realizing that, for $j_1<j_2$,}
      C^\mathrm{cons}_i(j_1,j_2) &=
      \begin{cases}
        C^\mathrm{cons}_i(j_1,j_2-1)+1&\text{ if }\exists
        M\in \mathcal{M}(i):\{\pi_i(j_2-1),\pi_i(j_2)\}\subseteq M \\
        C^\mathrm{cons}_i(j_1,j_2-1) & \text{ otherwise}
      \end{cases}
    \end{align}
    Clearly, the values $\operatorname{maxy}$ and
    $\operatorname{miny}$ can be computed in
    $\mathcal{O}(|\mathcal{C}|^2\cdot \ell)$ time.
    Define the directed acyclic graph $G$ with vertex set
    $\mathcal{C}'$ that has, for every $(c,c')\in \mathcal{C}_<^2$, an
    arc from $c$ to $c'$ if $\bigcap_{i\in \ell}
    [\operatorname{miny}_i(c,c'),\operatorname{maxy}_i(c,c')]
    \ne \emptyset$.
    Note that each arc corresponds exactly to a pair of characters
    that can be realized simultaneously without wiggle.
    Since $G$ is acyclic, we can compute a longest path~$P$ in~$G$
    in time $\mathcal{O}(|\mathcal{C}|^2)$.  Note that the characters
    of~$P$ represent a wiggle-free solution that consists of
    $|V(P)|+1$ characters from~$\mathcal{C}'$. A corresponding coordination can for example be computed by using the linear program from \cref{section:linearprogram} with additional equality constraints for characters that do not wiggle.
\end{proof}
The next result follows directly, as for $\ell=2$, the only characters
that can have (at most one) wiggle, are active at both
time steps.
\begin{corollary}
    \WiggleCProb\ is solvable in time $\mathcal{O}(|\mathcal{C}|^2)$ for $\ell=2$.
\end{corollary}

\subsection{NP-Hardness} %
\label{section:hardness}

In this section, we prove the following result with a gadget-based reduction. We first describe the source problem of the reduction, then the required gadgets, and finally the full reduction. 

\begin{theorem}\label{thm:wigglecounthardness}
    The decision variant of \WiggleCProb\ is \NP-complete.
\end{theorem}
\begin{figure}[b]
    \centering
    \includegraphics[page=1]{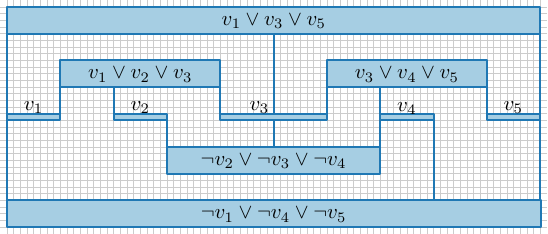}
    \caption{An instance of \Satprobshort.}
    \label{fig:inputinstance}
\end{figure}
Our proof reduces from the problem \Satprob, which is defined as
follows.  Let $\phi$ be a Boolean \textsc{3-Sat} formula in
conjunctive normal form, let~$\Upsilon$ be the set of variables
of~$\Phi$, and let $\Gamma$ be the set of clauses of~$\Phi$. %
The formula $\phi$ is \emph{monotone} if every clause of~$\phi$
contains either only positive or only negative literals.
We define the bipartite variable-clause graph $G(\phi)$, consisting of vertex set $\Gamma \cup \Upsilon$ and an edge between a clause $\gamma \in \Gamma$ and a variable $v \in \Upsilon$ if and only if $v$ occurs in $c$ (positively or negatively).
A \emph{monotone planar rectilinear embedding} $\mathcal{E}$ of $G(\phi)$ is a planar embedding of $G(\phi)$ on an integer grid (see \cref{fig:inputinstance})
such that
\begin{itemize}
    \item all vertices in $G(\phi)$ are represented by disjoint integer coordinate rectangles,
    \item the y-coordinates of all variable vertices are the same,
    \item all positive clause vertices are placed above the variable vertices,
    \item all negative clause vertices are placed below the variable vertices, and
    \item edges are straight vertical segments with integer x-coordinates connecting the borders of their corresponding rectangles.
\end{itemize}

\begin{problem}[\Satprob\ (\Satprobshort)]
Given a monotone \textsc{3-Sat} formula~$\phi$ and a monotone planar rectilinear embedding $\mathcal{E}$ of $G(\phi)$, decide whether $\phi$ is satisfiable.
\end{problem}

\Satprobshort is \NP-complete \cite{DBLP:journals/ijcga/BergK12}. 
We can assume w.l.o.g.\ that the coordinates in the embedding $\mathcal{E}$ have the following special properties:
\begin{itemize}
    \item The x-coordinates of  edges and vertical borders of the vertex rectangles are multiples of 8. 
    \item All edge lengths are multiples of 4.
    \item The clause vertices have height 4 and the variable vertices have height 1.
\end{itemize}
Notice that \cref{fig:inputinstance} shows an instance satisfying these properties.

Now, given such an instance $I_{\mathrm{SAT}}=(\phi, \mathcal{E})$ of \Satprobshort, we construct an instance $I_{\mathrm{SL}}=(\mathcal{C},T,\mathcal{M},A,(\pi_t)_{t\in [\ell]},\di,\da)$ of \WiggleCProb. We will give the precise number of wiggles $W$ for the decision variant of \WiggleCProb\ later, and set $\di=\da=1$, thus, we can assume that coordinations of $I_{\mathrm{SL}}$  contain only integer coordinates (see \cref{lemma:wigglecountinteger}).

To have a better correspondence between the two instances,  we assume in the reduction that the storyline curves are not x-monotone, but are y-monotone from bottom to top. That is, we rotate the storyline drawing by 90 degrees counterclockwise. 
All figures for the reduction reflect this.
We assume w.l.o.g.\ that the smallest y-coordinate of any vertex in $\mathcal{E}$ is~$3$. Let $y_{\max}$ be the largest y-coordinate.
We set $T=\{1,\dots, y_{\max}+2\}$. The main part of our reduction will be in $[3,y_{\max}]\subseteq T$, while the remaining four time steps will be used for an auxiliary bounding \emph{frame}.

Our reduction consists of three components: variable gadgets, clause gadgets, and a rigid frame.
For presentation purposes, we define the storyline instance graphically. That is, we define the instance in terms of one of its coordinations. The definitions of $\mathcal{C}, \mathcal{M}, A$, and $(\pi_t)_{t\in [\ell]}$ can then be inferred from this coordination.
In the presentation of the variable and clause gadget, we assume a black-box bounding frame, which we assume to have fixed coordinates in any coordination. Clearly, this frame must consist of meetings and characters.
An instance of \WiggleCProb\ that results from the instance in \cref{fig:inputinstance} can already be found in \cref{fig:fullexampleblackbox} in the variant with a black-box bounding frame and in \cref{fig:recomplete} with the full definition of the frame.  We start with the description of the variable gadget.

\subparagraph*{Variable gadget.}

Assume all variables in $\mathcal{E}$ are embedded on the two y-coordinates $y_{\mathrm{top}}$ and $y_\mathrm{bot}$.
Now consider a variable $v\in\Upsilon$. We assume that the corresponding rectangle in $\mathcal{E}$ has width $w_v$, and its horizontal edges span from $x_\mathrm{left}$ to $x_\mathrm{right}$. Then the corresponding variable gadget (see \cref{fig:vargadget}) 
\begin{figure}[tb]
    \centering
    \includegraphics[page=4]{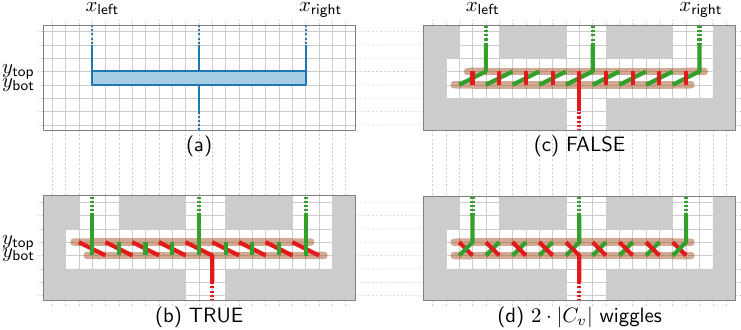}
    \caption{Illustration of a variable gadget. (a) The variable in the embedding $\mathcal{E}$. (b) The corresponding variable gadget in the TRUE state. (c) The variable gadget in the FALSE state. (d) The variable gadget in a state with too many wiggles. Character curves are green and red. meetings are brown horizontal blocks. The rigid frame is shown in gray. The dotted lines sketch the connections to the corresponding clause gadgets.}
    \label{fig:vargadget}
\end{figure}
consists of a set $C_v$ of $w_v$ characters which are ordered as in \cref{fig:vargadget} and are active at time steps $y_\mathrm{bot}$ and $y_\mathrm{top}$. This set of characters can be further partitioned into a set $C_v^\mathrm{pos}$ of \emph{positive characters} (green in \cref{fig:vargadget}) and a set $C_v^\mathrm{neg}$ of \emph{negative characters} (red in \cref{fig:vargadget}). Furthermore, using a frame as in the figure, the variable gadget is confined to a space of width $w_v+2$, with coordinates from $x_{\mathrm{left}}-1$ to $x_\mathrm{right}+1$.

For each occurrence of a variable in a clause, a character is extended towards a clause gadget, as indicated with the dashed line ends in \cref{fig:vargadget}. Positive characters are connected with positive clauses, negative characters with negative clauses. More precisely, for a vertical edge of $\mathcal{E}$ connecting $v$ with a positive (negative) clause $\gamma\in \Gamma$ at x-coordinate $x_\mathrm{left}+x_{v,\gamma}$ with $x_{v,\gamma}\in \mathbb{N}_0$, let $c_{v,\gamma}\in C_v^\mathrm{pos}$ ($c_{v,\gamma}\in C_v^\mathrm{neg}$) be the $(1+x_{v,\gamma}/2)$th (from left to right) positive (negative) character in $C_v^\mathrm{pos}$ ($C_v^\mathrm{neg})$ (the values are always positive as x-coordinates of edges are multiples of eight).
We call such a character \emph{wire character}.  This character is extended to time steps $y_\mathrm{top}+1,y_\mathrm{top}+2,\dots$ ($y_\mathrm{bot}-1,y_\mathrm{bot}-2,\dots$) until it reaches the corresponding clause gadget. This connection is along a corridor of width two defined by x-coordinates $x_\mathrm{left}+x_{v,\gamma}$ and $x_\mathrm{left}+x_{v,\gamma}+1$ (see \cref{fig:vargadget}). In a potential coordination, we say that $c_{v,\gamma}$ is \emph{satisfying} if the x-coordinate of $c_{v,\gamma}$ at $y_\mathrm{top}$ ($y_\mathrm{bot}$) is $x_\mathrm{left}+x_{v,\gamma}$, and not satisfying otherwise. We say that $x_\mathrm{left}+x_{v,\gamma}$ is the \emph{satisfying x-coordinate} of $c_{v,\gamma}$.

Due to the fixed coordinates of frames, we observe the following.
\begin{observation}
    The variable gadget has either $|C_v|$ (\cref{fig:vargadget}b,c) or $2|C_v|$ (\cref{fig:vargadget}d) wiggles between time steps $y_{\mathrm{top}}$ and $y_{\mathrm{bot}}$.
\end{observation}
Furthermore, as desired, variable gadgets model a choice between TRUE and FALSE (see \cref{fig:vargadget}b,c).
\begin{observation}\label{observation:vargadget}
    If the variable gadget has $|C_v|$ wiggles, then either all positive wire characters are satisfying and all negative wire characters are not satisfying, or vice versa.
\end{observation}

\subparagraph*{Clause gadget.}
\begin{figure}[htb]
    \centering
    \includegraphics[page=5]{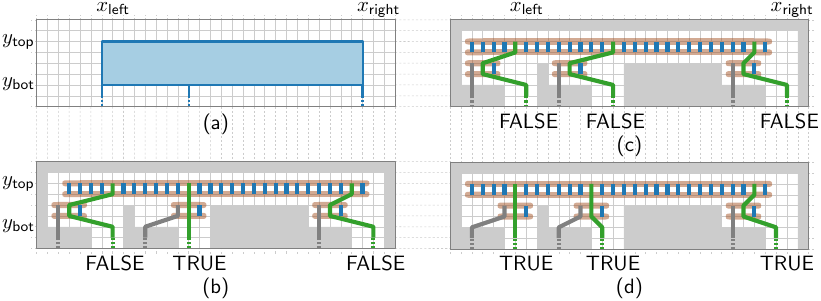}
    \caption{Illustration of a positive clause gadget. (a) The clause in the embedding $\mathcal{E}$. (b-d) The same clause gadget in three different states depending on whether its wire characters are satisfying (TRUE) or not (FALSE). A character belonging to the rigid frame is shown with the gray line.}
    \label{fig:clausegadget}
\end{figure}
At the heart of the reduction is the clause gadget, see \cref{fig:clausegadget}. We  desribe the  gadget for positive clauses only. The gadgets for negative clauses are the same but mirrored along the x-axis and can also be found in \cref{fig:fullexampleblackbox,fig:recomplete}.
Thus, let $\gamma$ be a positive clause with three variables $v_1,v_2,v_3\in \Upsilon$, such that $v_1$ is the left, $v_2$ is the middle, and $v_3$ is the right variable. Assume that the clause rectangle in $\mathcal{E}$ is defined by the coordinates $[x_\mathrm{left},x_\mathrm{right}]\times [y_\mathrm{bot},y_\mathrm{top}]$.
The clause gadget for $\gamma$ is essentially placed into the rectangle $[x_\mathrm{left}-4,x_\mathrm{right}+1]\times [y_\mathrm{bot},y_\mathrm{top}]$.\footnote{That is also partially why we need x-coordinates being multiples of eight.}
The rest of the clause gadget depends on the satisfying x-coordinates
of the characters $c_{v_i,\gamma}$ with $i\in[3]$; see
\cref{fig:clausegadget}. Thus, let $x^s_1,x^s_2,x^s_3$ be the
satisfying x-coordinates of
$c_{v_1,\gamma},c_{v_2,\gamma},c_{v_3,\gamma}$, respectively.
Each of the three wire characters is active until time step $y_\mathrm{top}$, and it is part of some meetings with \emph{clause gadget characters} (blue in \cref{fig:clausegadget}).
For each wire character $c_{v_i,\gamma}$, we have the following construction.
\begin{itemize}
    \item A \emph{fixing character} which is a character from the frame fixed at $x^s_i-4$ which is extended until $y_\mathrm{bot}+2$, see the gray line in \cref{fig:clausegadget}.
    \item The two \emph{blocking meetings} of size three at time steps $y_\mathrm{bot}+1$ and $y_\mathrm{bot}+2$ that ensure that $c_{v_i,\gamma}$ must have at least one wiggle if its x-coordinate at $y_\mathrm{bot}$ is not satisfying.
\end{itemize}
Lastly, the clause gadget has two \emph{choice meetings} of size $x_\mathrm{right}-x_\mathrm{left}+4$ at $y_\mathrm{top}$ and $y_\mathrm{top}-1$.
Each clause character $c_{v_i,\gamma}$ with $i\in[3]$ appears inside these meetings at position $x^s_i-x_\mathrm{left}+i$. The remaining positions are filled by clause characters which have the same position in both meetings. Essentially, the addition of $i$, which slightly shifts the relation of satisfying x-coordinate and position in the meeting, ensures that we cannot prevent multiple wiggles by having multiple satisfying wire characters.
The following two lemmas show that the clause gadget can be realized with the fewest wiggles if at least one of its characters is satisfying, and otherwise with one wiggle more.
We assume that wire characters do not have wiggles between their variable gadget and $y_\mathrm{bot}$, as they are either unnecessary or can be moved into the clause gadget.
\begin{lemma}\label{lemma:wigglesclausegadgetfalse}
  Assume that there exists a nice coordination in
  which none of $c_{v_1,\gamma}$, $c_{v_2,\gamma}$, and $c_{v_3,\gamma}$
  is satisfying. Then a nice coordination of the clause gadget has at
  least six wiggles, and such a coordination with exactly six wiggles
  always exists.
\end{lemma}
\begin{proof}
    \cref{fig:clausegadget} depicts such a nice coordination.
    We show that the number of wiggles of each wire character
    $c_{v_i,\gamma}$ with $i \in [3]$ and its corresponding fixing
    character is at least two between $y_\mathrm{bot}$ and
    $y_\mathrm{top}$.  If the number of wiggles is at least two between
    $y_\mathrm{bot}$ and $y_\mathrm{bot}+1$, we are done.

    Otherwise, $c_{v_i,\gamma}$ has x-coordinate $x^s_i-3$ at time step $y_\mathrm{top}-3$. Now notice that there is no nice coordination of the choice and blocking meetings such that $c_{v_i,\gamma}$ can be drawn without a wiggle between $y_\mathrm{top}-3$ and $y_\mathrm{top}-1$.
\end{proof}

\begin{lemma}\label{lemma:wigglesclausegadgettrue}
  Assume a nice coordination in which some of $c_{v_i,\gamma}$ with
  $i\in[3]$ is satisfying.  Then a nice coordination of the clause
  gadget has at least five wiggles, and such a coordination with
  exactly five wiggles always exists.
\end{lemma}
\begin{proof}
    For the existence, let $c_{v_j, \gamma}$ with $j\in [3]$ be some satisfying wire character. Simply draw $c_{v_j, \gamma}$ without wiggle, by adjusting the coordination of the choice meetings accordingly. The remaining two wire characters can be drawn with two wiggles, while their fixing character remains without wiggle between $y_\mathrm{bot}$ and $y_\mathrm{top}$ (see \cref{fig:clausegadget}a, where $j=2$).

    For the lower bound we have already seen in the proof of \cref{lemma:wigglesclausegadgetfalse} that if a fixing character remains without wiggle between $y_\mathrm{bot}$ and $y_\mathrm{top}$, then its corresponding wire character must have at least one wiggle between $y_\mathrm{top}-2$ and $y_\mathrm{top}-1$. Thus, for there to be at most 4 wiggles, two wire characters would need to be realized without wiggles. This is not possible because of the slight shift of positions of the wire characters in the choice meetings, a contradiction.
\end{proof}

\begin{figure}[tb]
    \centering
    \includegraphics[width=\textwidth,page=8]{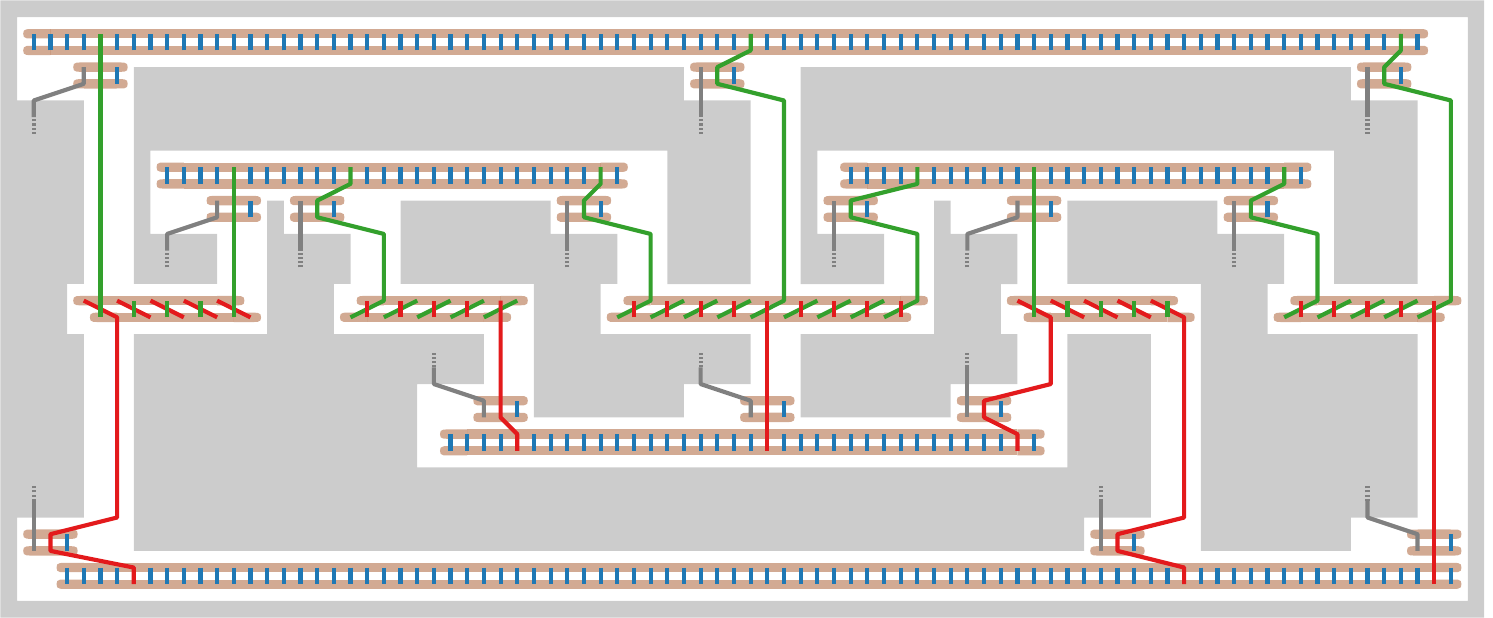}
    \caption{The complete reduction (with black-box frame) resulting
      from the %
      instance in \cref{fig:inputinstance}.}
    \label{fig:fullexampleblackbox}
\end{figure}

\subparagraph{Realizing the rigid frame}
We need to be careful in how we realize the rigid frame from \cref{fig:fullexampleblackbox}. The first idea is to define a character for each time step and each x-coordinate that we assumed to be part of the frame; character orderings are again derived directly from their x-coordinates. If we have a character at the same x-coordinate in two consecutive time steps, we actually use the same character for modeling the frame. Further, characters at consecutive x-coordinates and at the same time step will be part of the same meeting. This is indicated in \cref{fig:redincomplete}, which shows the complete reduction with this version of the frame. For each time step, frame characters which are in the same gray region, form a meeting.
\begin{figure}[tb]
    \centering
    \includegraphics[page=6,width=\textwidth]{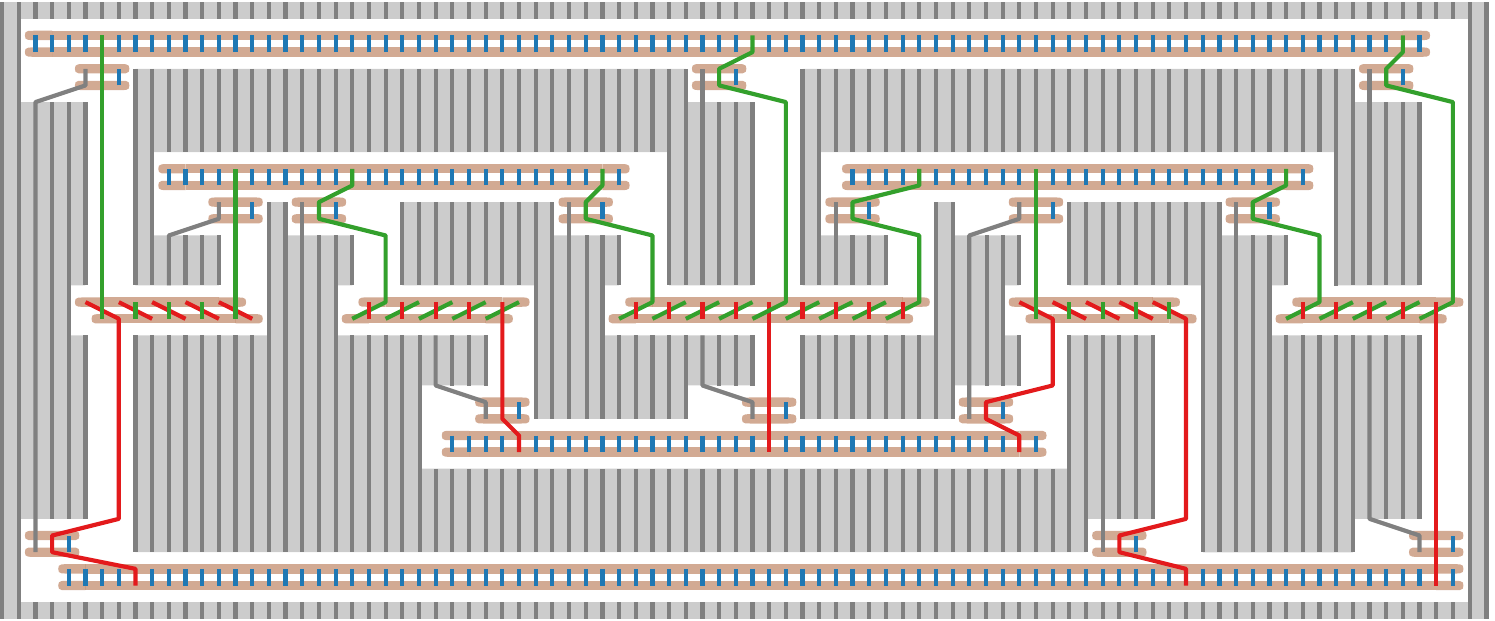}
    \caption{The complete reduction using the naive frame resulting from the \Satprobshort instance in \cref{fig:inputinstance}. Characters belonging to the frame are shown in gray. Meetings involving frame characters are not shown explicitly.}
    \label{fig:redincomplete}
\end{figure}
\begin{figure}[tb]
    \centering
    \includegraphics[page=7,width=\textwidth]{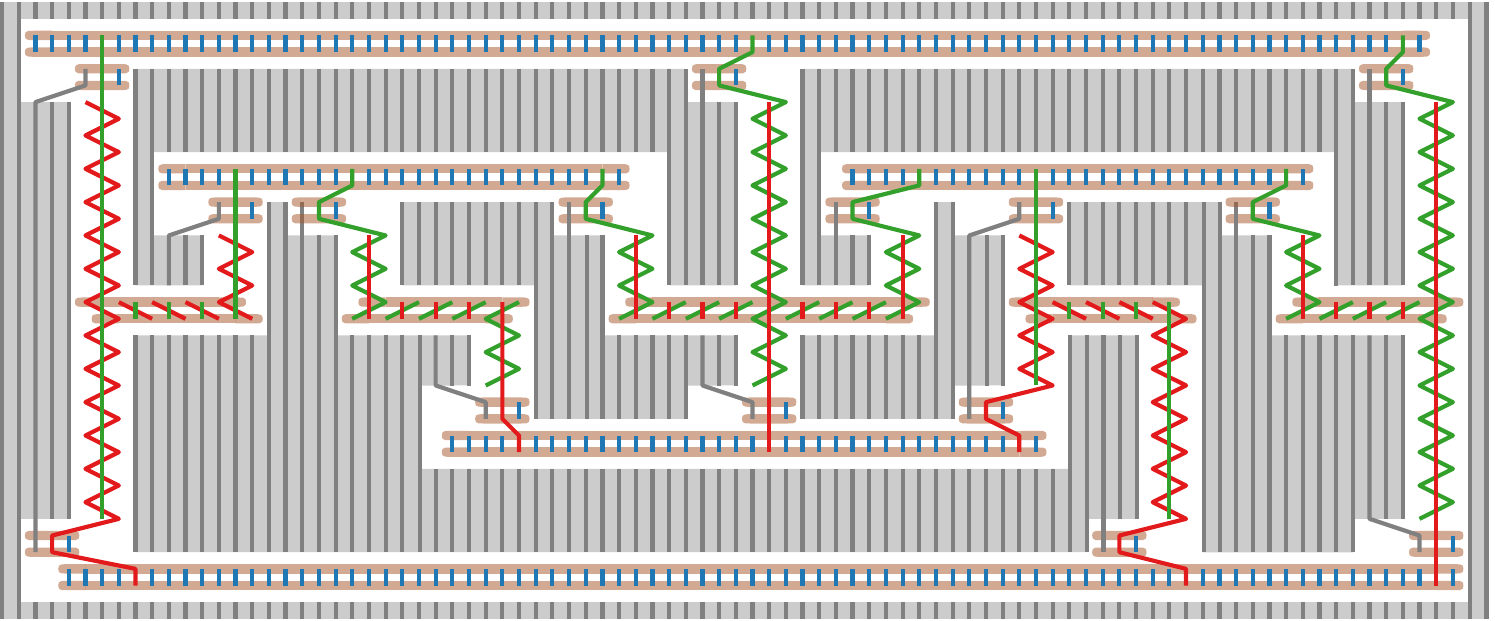}
    \caption{The complete reduction resulting from the \Satprob\
      instance in \cref{fig:inputinstance}. Meetings involving frame
      characters are not shown explicitly. The necessity of the zig-zag curves is described in the description of the rigid frame.}
    \label{fig:recomplete}
\end{figure}
However, there is a problem with this variant of the frame: the gray regions in \cref{fig:redincomplete} form disconnected components, and are thus not ``rigid enough'' to easily verify the correctness of the reduction. We can circumvent this by slightly adapting our construction, see \cref{fig:recomplete}.
Namely, for each wire character $c_{v,\mathrm{cl}}$, let $c'$ be the character in its variable gadget that is involved in a crossing with $c_{v,\mathrm{cl}}$. Then we extend $c'$ up or down until the border of the clause gadget, while the ordering of $c_{v,\mathrm{cl}}$ and $c'$ is alternating (see \cref{fig:recomplete}). Further, the pathways from variables to clauses are widened by one unit to the left.
Now, when considering two adjacent time steps, it is easy to see that the frame is indeed rigid (when assuming it does not incur too many wiggles). Indeed, if the frame was not rigid, there would need to be pathways of size two, and the clauses $c_{v,\mathrm{cl}}$ and $c'$ would have two wiggles instead of one between two consecutive time steps. 

We define $k_{\mathrm{cr}}$ as the number of crossings in the reduced storyline instance; in any coordination they will introduce a wiggle. Now, let $k=k_{\mathrm{cr}}+5|\Gamma|$. We claim the following.
\begin{restatable}{lemma}{reductioncorrectness}\label{lemma:reductioncorrectness}
    $I_{\mathrm{SAT}}$ is a yes-instance of \Satprob\ if and only if $I_{\mathrm{SL}}$ has a nice coordination with at most $k$ wiggles.
\end{restatable}
\begin{proof}    
    The forward direction follows immediately from our construction.

    For the backward direction, assume a nice coordination of $I_{\mathrm{SL}}$ with at most $k$ wiggles.
    It is clear that none of the frame characters has a wiggle; note that the technicalities are tedious to show, but one can convince oneself by blowing up the construction to have at least $5|\Gamma|+1$ characters in each meeting containing frame characters. Then, as $k_{\mathrm{cr}}$ crossings are unavoidable, clearly no frame character can wiggle. This can, for example, be achieved by assuming that differences between x-coordinates of the \Satprob\ instance are larger than eight, but depend on $\Gamma$; note that we chose eight to be able to visualize the construction.
    Furthermore, we can assume that the connections between variable and clause gadgets use corridors of width exactly three. Again, the technicalities are tedious to show, but one can convince oneself by blowing up the construction in $y$-direction; i.e., assuming that the connections between variables and clauses in the \Satprob\ instance are multiples of some value depending on $\Gamma$.
    Hence, with the above two observations, we can indeed assume that the frame is completely rigid.

    Now, because $k_\mathrm{cr}$ wiggles are unavoidable, and we have at least $5$ wiggles for each clause gadget by \cref{lemma:wigglesclausegadgetfalse} and \cref{lemma:wigglesclausegadgettrue}, we must have exactly $5$ wiggles per clause gadget. Further, each variable gadget has $|C_v|$ wiggles for each variable $v$.
    Hence, at least one wire character per clause gadget is satisfying, and we can use the variable assignment corresponding to satisfying characters to verify the satisfiability of the \Satprob\ instance by \cref{observation:vargadget}.
\end{proof}
\begin{proof}[Proof of \cref{thm:wigglecounthardness}]
    For \NP-membership we provide a polynomial certificate and a certifier for the decision problem of deciding whether there is a nice coordination with at most $k$ wiggles. The certificate consists of a set of pairs $X\subseteq \mathcal{C}\times [\ell-1]$ of size at least $\sum_{c\in \mathcal{C}}(|A(c)|-1)-k$. Each pair $(c,t)\in X$ defines that $c$ can be drawn without wiggle between time steps $t$ and $t+1$. The certifier can certify the certificate using a linear program. The linear program is essentially the one from \cref{sec:twh-lp} using constraints \eqref{cons:di}-\eqref{cons:G0} and the additional constraints $y_{c,t}=y_{c,t+1}$ for each pair $(c,t)\in X$. If the polytope defined by the program is non-empty, then the certifier reports yes.

    \NP-hardness follows from \cref{lemma:reductioncorrectness} and the fact that \Satprobshort is \NP-hard. Clearly, the reduction is polynomial.
\end{proof}

\subsection{ILP Formulation for \WiggleCProb}\label{sub:wc-ilp}

\newcommand*\ymax{\ensuremath{Y}} %

We have established that \WiggleCProb is \NP-complete.
We now present an ILP formulation for \WiggleCProb (assuming $\di,\da\in\mathbb{N}$) that we will use
in \cref{sec:case-studies} to construct storyline drawings with minimum wiggle count.
The formulation requires the following lemma.
\begin{lemma}\label{lemma:wigglecountinteger}
    Consider an ordered storylines instance and $\di,\da\in \mathbb{N}$. For each $k\in \mathbb{N}$, if there is a nice coordination with $k$ wiggles, there is one with at most $k$ wiggles using integer coordinates.
\end{lemma}
\begin{proof}
    The proof is similar to the one of \cref{thm:integerpolytope}. We recommend reading that proof first. 
    Consider a nice coordination with $k$ wiggles.
    We build a graph $G$ with $V(G)=\{y_{t,c}\mid t\in [\ell],c\in \Ac(t)\}$ and $E(G)=\{y_{t,c}y_{t',c'}\mid y_{t,c}-y_{t',c'}\in \mathbb{Z}\}$. If $G$ has multiple connected components, the y-coordinates corresponding to vertices in a single connected component can be increased until the number of connected components of $G$ is decreased. When $G$ is connected, all y-coordinates can be increased until they are all integers. 
\end{proof}

The ILP makes use of the following variables (refer to \cref{sec:twh-lp} for details).
\begin{itemize}
    \item $y_{t,c}$ for $t\in [\ell],c\in \Ac(t)$, y-coordinate of character $c$ at time step $t$.
    \item $z_{t,c}$ for $t\in [\ell-1], c\in \Ac(t)\cap \Ac(t+1)$,
        equals~$1$ if character~$c$ wiggles between~$t$ and~$t+1$,~$0$~otherwise.
\end{itemize}

By \cref{lemma:wigglecountinteger} we can assume the y-coordinates of the 
characters to be integers.  To enforce that $z_{t,c}=0$ if and only if
$y_{t,c}=y_{t+1,c}$ (i.e., there is no wiggle), we use a large constant~$\ymax \in \mathbb{N}$; 
see constraints \eqref{cons:m1} and \eqref{cons:m2} below.
A trivial upper bound for the y-coordinates is
$\max\{\di,\da\} \cdot \sum_{t\in [\ell]} |\Ac(t)|$.  We use this bound for~\ymax.
The sets $N_{\mathcal{M}}(t)$ and $N_A(t)$ are defined in \cref{sec:twh-lp}.
Now we can formulate the ILP:
\begin{align*}
    \text{Minimize } \quad  \sum_{t\in [\ell-1],c\in \Ac(t)} z_{t,c}                                                          \\
    \text{subject to } \hspace{14mm} %
    y_{t,c'}-y_{t,c}  & =\di                   &  & \text{for all } t\in [\ell], (c,c')\in N_{\mathcal{M}}(t)     \tag{Cons-\di}\label{cons:da*} \\
    y_{t,c'}-y_{t,c}  & \ge \da                &  & \text{for all } t\in [\ell], (c,c')\in N_A(t)     \tag{Cons-\da}\label{cons:di*} \\
    y_{t,c}+\ymax\cdot z_{t,c} & \ge y_{t+1,c} &  & \text{for all } t \in [\ell-1],c\in \Ac(t)\cap \Ac(t+1)\tag{M1}\label{cons:m1} \\
    y_{t,c}-\ymax\cdot z_{t,c} & \le y_{t+1,c} &  & \text{for all } t \in [\ell-1],c\in \Ac(t)\cap \Ac(t+1)\tag{M2}\label{cons:m2} \\
    y_{t,c} \in\mathbb{N}_0,~~0 \le y_{t,c} & < \ymax &  & \text{for all }t\in [\ell],c\in\Ac(t)\tag{Int}\label{cons:int}            \\
    z_{t,c}           & \in \{0,1\}            &  & \text{for all }t\in [\ell],c\in\Ac(t)\tag{Bin}\label{cons:bin}
\end{align*}
Constraints \eqref{cons:da*} and \eqref{cons:di*} have been explained in \cref{sec:twh-lp}.
\eqref{cons:int} and \eqref{cons:bin} ensure that $y_{t,c}$ is a positive integer bounded by \ymax{} and
$z_{t,c}$ is binary, respectively.
The objective sums over all $z$-variables and therefore counts the number of wiggles.

As it turned out, optimal solutions created by state-of-the-art solvers produce unusably
large wiggles with this formulation.  We thus introduce a new variable~$h$
for the height of the drawing and, for each $t\in [\ell]$ and $c\in \Ac(t)$, the new constraint $y_{t,c} \le h$.
By adding the term $h/\ymax$ to the objective, we also minimize~$h$.  The adjusted formulation still solves \WiggleCProb.

\section{Routing Character Curves}
\label{sec:routing}

\newcommand*\oct{\ensuremath{\overline{c_t}}}
\newcommand*\rtc{\ensuremath{r_{t,c}}}
\newcommand*\rtcc{\ensuremath{r_{t,\hat{c}}}}
\newcommand*\rmin{\ensuremath{r_{\min}}}
\newcommand*\dxt{\ensuremath{dx_t}}
\newcommand*\dytc{\ensuremath{dy_{t,c}}}
\newcommand*\raddist[3]{\ensuremath{\rule{0pt}{1em}\overrightarrow{\rho_{#1,#2}\rule{0pt}{1.3ex}}(#3)}}
\newcommand*\nupl{\ensuremath{N_{\uparrow,\leftarrow}}}
\newcommand*\nupr{\ensuremath{N_{\uparrow,\rightarrow}}}
\newcommand*\ndnl{\ensuremath{N_{\downarrow,\leftarrow}}}
\newcommand*\ndnr{\ensuremath{N_{\downarrow,\rightarrow}}}
\newcommand*\rthc{\ensuremath{r_{t,\hat{c}}}}

In our storyline drawings, we represent each wiggle by a sequence of two circular arcs.  In order to generate a smooth
transition, we require the tangents at the arcs to be equal where the arcs cross over
(see \cref{fig:geo-1}).
In order to achieve visually pleasing transitions, we require that all arcs have at least
a certain radius and that consecutive lines wiggle “in parallel”
(a precise definition is given later in this section).

\begin{figure}[b!]
    \begin{subfigure}[b]{.32\linewidth}
        \centering
        \includegraphics[page=2]{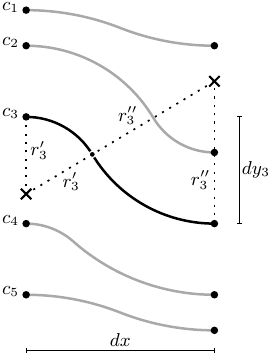}
        \subcaption{Wiggle geometry.}
        \label{fig:geo-1}
    \end{subfigure}
    \hfill
    \begin{subfigure}[b]{.32\linewidth}
        \centering
        \includegraphics[page=4]{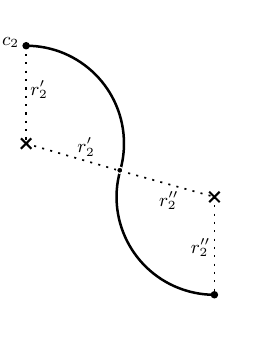}
        \subcaption{\nolinenumbers Not x-monotone.}
        \label{fig:geo-3}
    \end{subfigure}
    \hfill
    \begin{subfigure}[b]{.32\linewidth}
        \centering
        \includegraphics[page=5]{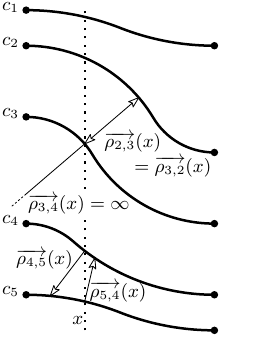}
        \subcaption{\nolinenumbers Radial curve distance.}
        \label{fig:geo-4}
    \end{subfigure}
    \caption{A wiggle is represented by a sequence of two circular arcs.}
    \label{fig:geo-layer-bad}
\end{figure}

When visualizing the transition between a pair of consecutive time steps~$t$ and~$t+1$,
for each character~$c$, the wiggle height $\dytc = |y_{t,c} - y_{t+1,c}|$ is given by the coordination.
Given the radii $r_{t,c}'$ and $r_{t,c}''$ of the two arcs of a character curve~$c$ between time steps~$t$
and~$t+1$, the required horizontal space \dxt{} for this configuration fulfills
\begin{restatable}{equation}{equationdx}
    \dxt^2 = 2 (\rtc' + \rtc'') \dytc - \dytc^2.
\end{restatable}
Refer to \cref{apx:geo} for a detailed derivation.
We require each character curve to be x-monotone.  If we choose \dxt{} badly, we may end up
with a sequence of arcs that is not x-monotone (see \cref{fig:geo-3}).  It is easy to see
that a character curve is x-monotone if and only if $(\rtc' + \rtc'') \geq \dytc$.
With \cref{thmt@@equationdx}, we can state this as $\dxt^2 \ge \dytc^2$.

Between two time steps, two character curves run in parallel, if they do not cross and
start and end with the same vertical distance.
In order to have them wiggle “in parallel”, we use concentric arcs.
Generally, for two x-monotone differentiable curves~$\varphi$ and~$\varphi'$, we define their
\emph{directed radial distance} at position~$x$, denoted by $\raddist{\varphi}{\varphi'}{x}$,
as the length of a normal line segment starting at~$\varphi(x)$ and ending where it
intersects~$\varphi'$ or infinity, if the normal line does not intersect~$\varphi'$.
See \cref{fig:geo-4} for an example.
We define the \emph{radial distance} of~$\varphi$ and~$\varphi'$ at~$x$ as
$\rho_{\varphi,\varphi'}(x) = \min(\raddist{\varphi}{\varphi'}{x}, \raddist{\varphi'}{\varphi}{x})$
(see \cref{fig:geo-4}).

We call two x-monotone curves \emph{related} if their axis-aligned bounding boxes overlap.
We call them \emph{upwards (downwards) co-oriented} if they are also y-monotone
and their y-coordinates both increase (decrease).
Let~$\varphi$ and~$\varphi'$ be two character curves 
between consecutive time steps~$t$ and~$t+1$.
By construction, $\varphi$~and~$\varphi'$ are y-monotone.
If they are also related, co-oriented, and do not cross each other,
we require their radial distance to be monotone.

We define the set $\nupl(t)$ of all pairs of
characters~$c,\hat{c}\in\Ac(t)\cap\Ac(t+1)$ such that
\begin{itemize}
    \item $y_{t,c} < y_{t,\hat{c}}$,
    \item $c$~and~$\hat{c}$ are non-crossing, related, and upwards co-oriented, and
    \item $y_{t,\hat{c}}-y_{t,c} \le y_{t+1,\hat{c}}-y_{t+1,c}$.
\end{itemize}
We define $\ndnl(t)$ for downwards co-oriented characters similarly as well as
$\nupr(t)$ and $\ndnr(t)$ if $y_{t,\hat{c}}-y_{t,c} >y_{t+1,\hat{c}}-y_{t+1,c}$.
Note that these sets contain only characters that wiggle
(i.e. their y-coordinates are different in $t$ and $t+1$).
We use the following LP
to determine~$\dxt^2$ as well as
the radii~$\rtc'$ and~$\rtc''$ for each time step~$t\in[\ell-1]$ and for each
$c \in \Ac(t) \cap \Ac(t+1)$.
\begin{align*}
    \text{Minimize } \quad  \dxt^2                                          \\
    \text{subject to }\quad
    \dxt^2        & \ge \max_{c \in \Ac(t)}\dytc^2    \tag{XM}\label{cons:xm} \\
    \dxt^2        & = 2(\rtc'+\rtc'')\dytc-\dytc^2 && \text{for all } c \in \Ac(t) \cap \Ac(t+1) \tag{R1}\label{cons:r1} \\
    \rtc', \rtc'' & \ge \rmin                      && \text{for all } c \in \Ac(t) \cap \Ac(t+1) \tag{R2}\label{cons:r2} \\
    \rtc'         & \ge \rthc'+(y_{t,c}-y_{t,\hat{c}})  && \text{for all } c, \hat{c} \in \ndnl(t) \tag{D1}\label{cons:d1} \\
    \rtc''        & \le \rthc''-(y_{t+1,c}-y_{t+1,\hat{c}}) && \text{for all } c, \hat{c} \in \ndnr(t) \tag{D2}\label{cons:d2} \\
    \rtc'         & \le \rthc'-(y_{t,c}-y_{t,\hat{c}})  && \text{for all } c, \hat{c} \in \nupl(t) \tag{D3}\label{cons:d3} \\
    \rtc''        & \ge \rthc''+(y_{t+1,c}-y_{t+1,\hat{c}}) && \text{for all } c, \hat{c} \in \nupr(t) \tag{D4}\label{cons:d4} 
\end{align*}

The constraint~\eqref{cons:xm} ensures x-monotonicity between time steps~$t$ and~$t+1$;
\eqref{cons:r1} ensures that the two arcs meet at a distinct point and their tangents match.
The remaining quality criteria are addressed by \eqref{cons:r2} (minimum radius) and
\eqref{cons:d1} to \eqref{cons:d4} (monotone radial distance). 
Note that we can model each variable of kind $\dxt^2$ as a linear variable because we do
not have any terms depend on $\dxt$ (non-squared).

\section{Case Studies}
\label{sec:case-studies}

\definecolor{PKdarkblue}{rgb}{0.121,0.47,0.705}
\definecolor{PKdarkred}{rgb}{0.89 0.102 0.109}
\definecolor{PKdarkgreen}{rgb}{0.2 0.627 0.172}
\definecolor{PKdarkorange}{rgb}{1 0.498 0}
\definecolor{PKdarkpurple}{rgb}{0.415 0.239 0.603}

\definecolor{PKlightblue}{rgb}{0.651 0.807 0.89}
\definecolor{PKlightred}{rgb}{0.984 0.603 0.6}
\definecolor{PKlightgreen}{rgb}{0.698 0.874 0.541}
\definecolor{PKlightorange}{rgb}{0.992 0.749 0.435}

\definecolor{PKlightgray}{rgb}{0.8 0.8 0.8}

\pgfplotscreateplotcyclelist{scatter-marks}{
  {mark options={solid,scale=1,fill=PKlightgray,draw=PKlightgray},mark=pentagon*,PKlightgray},
  {mark options={solid,scale=0.8,fill=PKdarkorange,draw=PKdarkorange},mark=*,PKlightgray},
  {mark options={solid,scale=1.2,fill=PKdarkgreen,draw=PKdarkgreen},mark=triangle*,PKlightgray},
  {mark options={solid,scale=1.2,fill=PKdarkblue,draw=PKdarkblue},mark=diamond*,PKlightgray},
}

\newcommand*\tsub[2]{\textsf{#1\textsubscript{#2}}}

We conducted two case studies that showcase wiggle minimization and the geometric realization we proposed.
Firstly, we compare the different metrics for storyline wiggle (see \cref{sec:preliminaries}) using
both previous datasets from the literature and new instances.  Secondly, we present a new use case for storylines
when visualizing rolling stock schedules.

\begin{figure}[ht]
    \centering
    \includegraphics{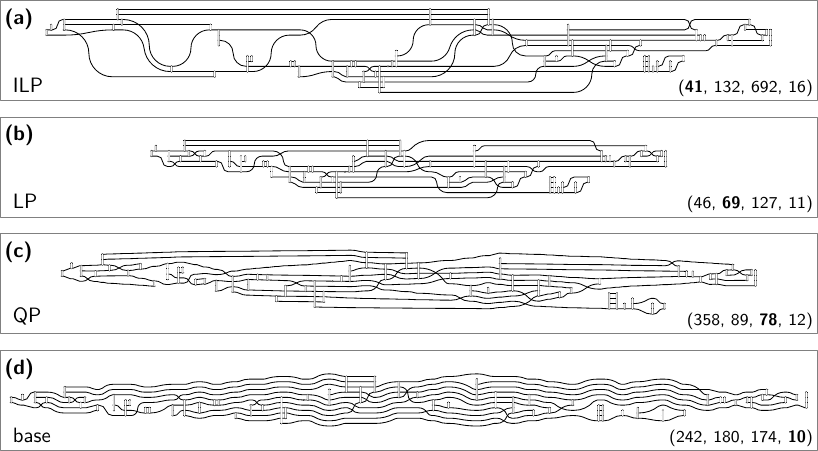}
    \caption{Storyline visualizations of the first chapter of Anna Karenina.  The values in parentheses are
        (wiggle count, linear wiggle height, quadratic wiggle height, height of the bounding box).}
    \label{fig:case-study-anna}
\end{figure}

\subsection{Storyline Benchmark Data}
\label{sub:benchmark}

We compiled a diverse set of benchmark instances for which we could create crossing minimal ordered storylines
using the ILP from \cite{DoblerJJMMN-GD24}.  It consists of three novels (\tsub{n}{1} to \tsub{n}{3}) originally
from The Stanford GraphBase \cite{knuth93} in a modified version by \cite{gronemann2016}, eight blockbuster
movies (\tsub{b}{1} to \tsub{b}{5} from \cite{DoblerJJMMN-GD24} and \tsub{b}{6} to \tsub{b}{8} from
\cite{tanahashi_efficient_2015}), one instance of publication data (\tsub{p}{1}) from \cite{dnsvw-cmtis-EuroCG23},
and five rolling stock schedules (\tsub{t}{1} to \tsub{t}{5}; see \cref{sub:rolling}).

We implemented the LP solving \WiggleHProb (see \cref{section:linearprogram}),
the QP solving \QdrWiggleHProb (see \cref{sub:quadratic}),
and the ILP for \WiggleCProb (see \cref{sub:wc-ilp})
including the additional secondary objective that minimizes drawing height.
We refer to them as \textsf{LP}, \textsf{QP}, and \textsf{ILP}, respectively.
We implemented \textsf{base} -- a heuristic that simply centers the character curves vertically in each time step -- as a baseline.
Note that \textsf{base} produces drawings of minimum height.

Refer to \cref{tab:dataset} for basic properties and detailed results for the instances in our benchmark set
where $|\mathcal{C}|$ is the number of characters, $|T|$ the number of time steps, and $|\mathcal{M}|$
the number of meetings (see \cref{sec:preliminaries}).  The center columns show the respective optimum
value for each metric, wiggle count (WC) as optimized by \textsf{ILP}, linear wiggle height (LWH) as
optimized by \textsf{LP}, quadratic wiggle height (QWH) as optimized by \textsf{QP}, and total bounding
box height (TH) as optimized by \textsf{base}.  The running times refer solely to solving the respective
mathematical program using Gurobi 12.0.2 under Arch Linux (Kernel 6.14)
running on a machine with an AMD Ryzen 7 7800X3D and 48\,GB of memory. 

\begin{table}[p]
\caption{Basic properties, optimal metrics, and running times for our benchmark set.}
\begin{tabular}{l @{~} l @{~}|@{~} r @{~~~} r @{~~} r | r @{~} r @{~} r @{~} r | r @{~} r @{~~} r}
\toprule
\multicolumn{5}{l|}{instance} & \multicolumn{4}{l|}{respective optimum} & \multicolumn{3}{l}{running time [s]} \\
            & & $|\mathcal{C}|$ & $|T|$ & $|\mathcal{M}|$ & WC & LWH & \multicolumn{1}{r}{QWH} & TH & \textsf{LP} & \textsf{QP} & \textsf{ILP} \\
\midrule
\tsub{n}{1} & Anna Karenina Ch.\,1     &  41  &  58  &     58 &  41 &  69 &    78.4 & 10   & <\,0.1 & <\,0.1 &      0.7 \\
\tsub{n}{2} & Huckleberry Finn         &  74  & 107  &    107 &  77 & 204 &   279.6 & 18   & <\,0.1 & <\,0.1 &     11.7 \\
\tsub{n}{3} & Les Misérables Ch.\,2--4 &  47  & 234  &    234 & 144 & 377 &   715.6 & 28   &    0.2 &    0.2 & 32,434.6 \\
\midrule                                                                                                       
\tsub{b}{1} & Back to the Future       &  34  &  67  &     67 &  58 & 160 &   202.8 & 19   & <\,0.1 & <\,0.1 &      6.7 \\
\tsub{b}{2} & Avatar (2009)            &  35  &  54  &     54 & 213 & 546 & 1,588.6 & 21   & <\,0.1 &    0.1 &  3,429.0 \\
\tsub{b}{3} & Ocean's Eleven           &  60  &  96  &     96 & 223 & 749 & 2,177.6 & 22   & <\,0.1 &    0.1 &     99.7 \\
\tsub{b}{4} & Harry Potter 1           &  52  &  54  &     54 & 179 & 661 & 2,111.1 & 33   & <\,0.1 &    0.1 &    725.4 \\
\tsub{b}{5} & Forrest Gump             & 113  & 103  &    103 &  68 & 214 &   248.3 & 20   & <\,0.1 & <\,0.1 &    238.4 \\
\tsub{b}{6} & The Matrix               &  14  & 101  &    278 &  37 &  48 &    39.1 & 12   & <\,0.1 & <\,0.1 &      0.1 \\
\tsub{b}{7} & Star Wars (1977)         &  14  & 200  &    824 &  51 &  76 &   146.2 & 10   & <\,0.1 &    0.1 &      0.8 \\
\tsub{b}{8} & Inception                &  10  & 490  &  1,466 &  51 &  78 &   116.4 & 7    & <\,0.1 &    0.1 &      0.7 \\
\midrule                                                                                                       
\tsub{p}{1} & gdea20                   &  19  & 100  &    100 &  46 &  96 &   145.4 & 16   & <\,0.1 &    0.1 &      9.2 \\
\midrule                                                                                                       
\tsub{t}{1} & DDZ                      &  38  & 113  &  1,976 &  42 &  84 &    96.6 & 36   & <\,0.1 &    0.2 &     49.1 \\
\tsub{t}{2} & FLIRT                    &  47  & 106  &  2,339 & 149 & 294 &   296.3 & 46   &    0.1 &    0.2 &    401.8 \\
\tsub{t}{3} & SGM                      &  69  & 173  &  5,262 & 263 & 529 &   540.0 & 68   &    0.2 &    0.3 &    184.0 \\
\tsub{t}{4} & SLT                      & 113  & 375  & 21,745 & 254 & 515 &   540.2 & 112  &    0.9 &    1.4 &  3,280.6 \\
\tsub{t}{5} & VIRM  & 146  & 241  & 17,410 & 170\textsuperscript{a}\hspace{-1ex} & 347 &   431.6 & 145  &    0.9 &    1.1 & \textit{dnf}\\
\bottomrule
\multicolumn{12}{l}{\footnotesize\textsuperscript{a}Wiggle count minimization did not finish within 24\,hours (best bound 165, gap 2.94\,\%).}
\end{tabular}
\label{tab:dataset}
\end{table}

Drawings of \tsub{n}{1} created with each of the four methods can be found in \cref{fig:case-study-anna}.
We see that minimizing the number of wiggles comes at a cost of some large wiggles and
that producing many -- thus small -- wiggles results in an unnecessarily wide drawing when using our style.
There could be a case for \textsf{QP} with an adjusted drawing style but for our current rendering, \textsf{LP} produces
the best results.

\begin{figure}[p]
    \centering
    \pgfplotsset{width=6cm,height=3cm,xmin=1,xmax=17,enlarge x limits=0.05,
        cycle list name=scatter-marks,unbounded coords=jump,
        xtick distance=1,x tick label style={anchor=base,yshift=-2ex},
        xticklabels={,,\tsub{n}{1},\tsub{n}{2},\tsub{n}{3},\tsub{b}{1},\tsub{b}{2},\tsub{b}{3},\tsub{b}{4},\tsub{b}{5},\tsub{b}{6},\tsub{b}{7},\tsub{b}{8},\tsub{p}{1},\tsub{t}{1},\tsub{t}{2},\tsub{t}{3},\tsub{t}{4},\tsub{t}{5}}}
    \hspace{3mm} %
    \begin{tabular}{r@{\hspace{10mm}}l}
    \begin{tikzpicture}[trim axis left]
    \begin{semilogyaxis}[ytick={1,2,4,8,16,32},yticklabels={1,2,4,8,16,32},
        legend columns=-1,legend entries={\textsf{base},\textsf{ILP},\textsf{LP},\textsf{QP}},
        legend to name=mylegend]
    \addplot+ table[y=base] {./graphics/m-wc.dat};
    \addplot+ table[y=ilp]  {./graphics/m-wc.dat};
    \addplot+ table[y=lp]   {./graphics/m-wc.dat};
    \addplot+ table[y=qp]   {./graphics/m-wc.dat};
    \draw (rel axis cs:0.03,0.95) node[anchor=north west] {wiggle count};
    \end{semilogyaxis}
    \end{tikzpicture}
        &
    \begin{tikzpicture}[trim axis left]
    \begin{semilogyaxis}[ytick={1,2,4,8,16},yticklabels={1,2,4,8,16}]
    \addplot+ table[y=base] {./graphics/m-lwh.dat};
    \addplot+ table[y=ilp]  {./graphics/m-lwh.dat};
    \addplot+ table[y=lp]   {./graphics/m-lwh.dat};
    \addplot+ table[y=qp]   {./graphics/m-lwh.dat};
    \draw (rel axis cs:0.03,0.95) node[anchor=north west] {linear wiggle height};
    \end{semilogyaxis}
    \end{tikzpicture}
        \\
    \begin{tikzpicture}[trim axis left]
    \begin{semilogyaxis}[ymax=64,ytick={1,2,4,8,16,32},yticklabels={1,2,4,8,16,32}]
    \addplot+ table[y=base] {./graphics/m-qwh.dat};
    \addplot+ table[y=ilp]  {./graphics/m-qwh.dat};
    \addplot+ table[y=lp]   {./graphics/m-qwh.dat};
    \addplot+ table[y=qp]   {./graphics/m-qwh.dat};
    \draw (rel axis cs:0.03,0.95) node[anchor=north west] {quadratic wiggle height};
    \end{semilogyaxis}
    \end{tikzpicture}
        &
    \begin{tikzpicture}[trim axis left]
    \begin{semilogyaxis}[ymax=3,ytick={1,1.5,2.25},yticklabels={1,1.5,2.25}]
    \addplot+ table[y=base] {./graphics/m-th.dat};  \label{pgflabel:m:base}
    \addplot+ table[y=ilp]  {./graphics/m-th.dat};  \label{pgflabel:m:ilp}
    \addplot+ table[y=lp]   {./graphics/m-th.dat};  \label{pgflabel:m:lp}
    \addplot+ table[y=qp]   {./graphics/m-th.dat};  \label{pgflabel:m:qp}
    \draw (rel axis cs:0.03,0.95) node[anchor=north west] {total height};
    \end{semilogyaxis}
    \end{tikzpicture}
    \end{tabular}
    \ref{mylegend}
    \caption{Metrics (relative to the respective optimum) of the
      wiggle minimization methods
      for our storyline benchmark set (\textbf{n}ovels,
      \textbf{b}lockbusters, \textbf{p}ublications, \textbf{t}rains).
      The y-axes are logarithmic.}
  \label{fig:metrics}
\end{figure}
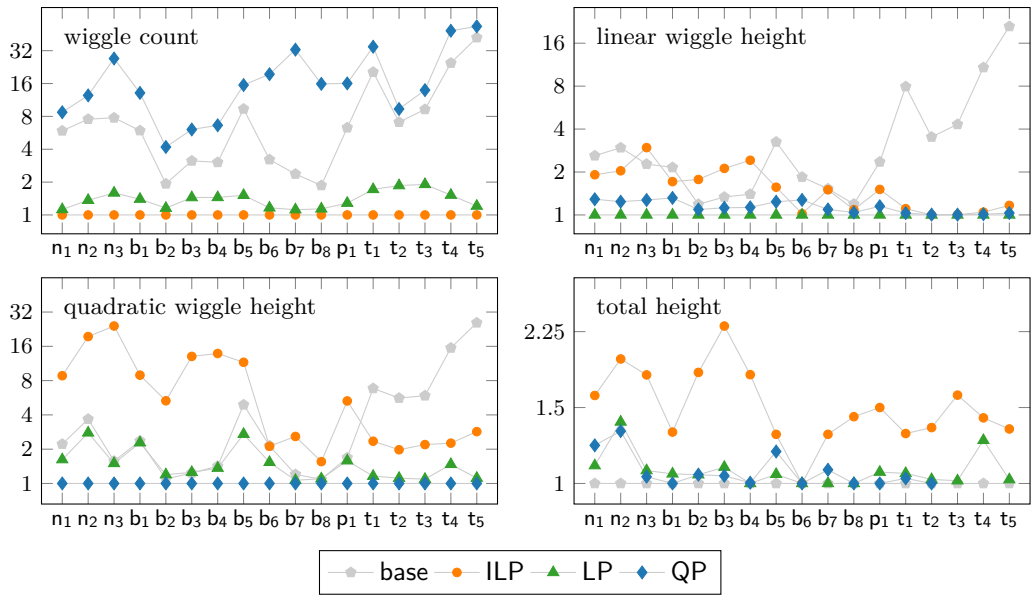

Detailed results of the performance metrics for the dataset can be found in \cref{fig:metrics}.  We see that optimizing for one objective
can lead to poor performance in other metrics.
Overall, however, \textsf{LP} produces good results
more consistently than the other methods.  This is why we used it for the drawings presented in our second case study.

\subsection{Rolling Stock Scheduling}
\label{sub:rolling}

To our own surprise, storylines show up in a seemingly unrelated area, namely the visualization of rolling stock schedules.  
Rolling stock scheduling, which is an important subtask in railway optimization, is defined as follows.
Given a set of train units,
a set of possible compositions of train units, a set of trips given by the timetable, and passenger demands per trip, the task is to find a feasible assignment of train units (possibly in composition) to trips that minimizes a weighted sum of operating costs.  
These costs include, e.g., the total distance driven,
the load factor, and the number of coupling operations.
There are different approaches for visualizing rolling stock schedules
(see~\cite{borndorfer_handouts_2019} and \cite[Fig.~1]{grimm_rolling-stock_2025}
for examples).

We visualize such a schedule as a storyline where each character curve represents the movement
of a specific train unit over the planning period (usually one day of operation)
and the composition of (typically two or three) train units is modeled as a meeting.
In contrast to our basic model, such meetings now extend over several layers.
Note that we use the term layer to avoid confusion with time steps in the timetable.
In order to satisfy the additional constraints of ``rolling stock'' storylines, we needed to extend our basic model.
We split each trip into at least two meetings, one at the departure time and one at the arrival time.
We then group as many consecutive meetings (of pairwise different train units)
into one layer as possible so the storyline stays compact.  If departure and arrival meeting of one trip
do not fall into consecutive layers, we insert an additional meeting for each intermediate layer.

Furthermore, we had to slightly adjust the ILP for crossing minimizing.
We fixed the ordering of the characters
within a meeting as they resemble the order of train units in a composition, which is relevant for planning.
Additionally, we could have required that meetings belonging to the same trip must be drawn straight
(as rectangles instead of as sequences of vertical bars) and that meetings cannot take part in a crossing,
but for simplicity, we have not enforced this.  However, we rewarded such solutions in the LP
for linear wiggle height minimization using soft constraints.  In the same way, we rewarded the vertical alignment of
meetings where a train composed of the same units operates another trip.

\begin{figure}[tb]
    \centering
    \includegraphics[width=\textwidth,trim={25mm 450mm 40mm 25mm}]{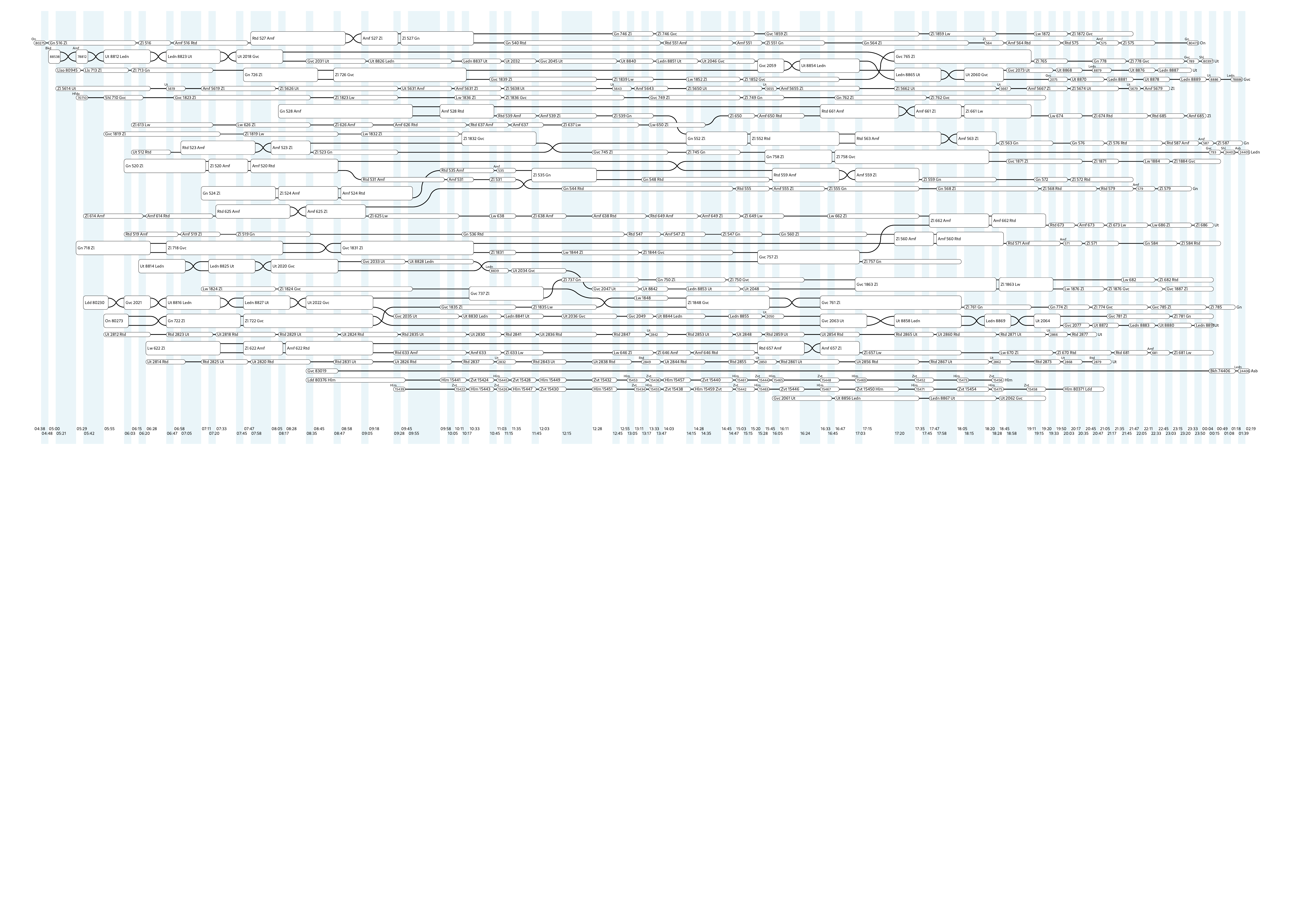}
    \caption{Storyline visualization of a rolling stock schedule using our drawing style.}
    \label{fig:case-study-ddz}
\end{figure}

An example of our visualization is shown in \cref{fig:case-study-ddz} (larger versions are in \cref{apx:fullscreen}).  
It depicts the same data set as \cite[Fig.~1]{grimm_rolling-stock_2025} (but unfortunately a different solution).
Because of how we assign meetings to layers,
the x-coordinates represent the departure and arrival times only vaguely.  
Also the width of a box does not represent the duration of a trip  exactly.  
This was a deliberate (but optional) choice to achieve a more compact visualization.  

\section{Conclusion}

We have given exact solutions for three variants of wiggle
minimization.  In particular, our efficient LP solution of \WiggleHProb
noticeably improves the aesthetic quality, while it does not cost much in
terms of computation time (compared to crossing minimization).

We leave several questions open.  
\begin{itemize}
    \item We have shown that \WiggleCProb is
\NP-complete, but we can solve the restriction to two time steps
efficiently.  
Is \WiggleCProb in \FPT\ with respect to the number of time
steps?
\item Our geometric routing produces very readable drawings, but
sometimes they are quite wide.  Is there a more compact geometric
solution that maintains the advantages of our method?
\item Despite the rich literature on optimizing quality metrics in storyline visualization, there exists no formal user evaluation on the effects on readability of their optimization. In particular, do visualizations of rolling stock schedules as crossing- and wiggle-optimized storyline visualizations lead to improvements in train operation planning?
\end{itemize}

\bibliography{abbrv,literature}

\begin{thebibliography}{10}

\bibitem{arendt_y_2017}
Dustin Arendt and Meg Pirrung.
\newblock The ``y'' of it matters, even for storyline visualization.
\newblock In Brian~D. Fisher, Shixia Liu, and Tobias Schreck, editors, {\em
  12th {IEEE} Conf. Visual Analytics, Science \& Technology (VAST)}, pages
  81--91, 2017.
\newblock \href {https://doi.org/10.1109/VAST.2017.8585487}
  {\path{doi:10.1109/VAST.2017.8585487}}.

\bibitem{dibartolomeoThereMoreStreamgraphs2016}
Marco~Di Bartolomeo and Yifan Hu.
\newblock There is more to streamgraphs than movies: Better aesthetics via
  ordering and lassoing.
\newblock {\em Comput. Graph. Forum}, 35(3):341--350, 2016.
\newblock \href {https://doi.org/10.1111/CGF.12910}
  {\path{doi:10.1111/CGF.12910}}.

\bibitem{borndorfer_handouts_2019}
Ralf Borndörfer, Boris Grimm, Markus Reuther, and Thomas Schlechte.
\newblock Optimization of handouts for rolling stock rotations.
\newblock {\em J. Rail Transport Planning \& Management}, 10:1--8, 2019.
\newblock \href {https://doi.org/10.1016/j.jrtpm.2019.02.001}
  {\path{doi:10.1016/j.jrtpm.2019.02.001}}.

\bibitem{bv-co-04}
Stephen Boyd and Lieven Vandenberghe.
\newblock {\em Convex Optimization}.
\newblock Cambridge University Press, 7th edition, 2009.

\bibitem{byronStackedGraphsGeometry2008}
Lee Byron and Martin Wattenberg.
\newblock Stacked graphs - geometry {\&} aesthetics.
\newblock {\em {IEEE} Trans. Vis. Comput. Graph.}, 14(6):1245--1252, 2008.
\newblock \href {https://doi.org/10.1109/TVCG.2008.166}
  {\path{doi:10.1109/TVCG.2008.166}}.

\bibitem{DBLP:journals/ijcga/BergK12}
Mark de~Berg and Amirali Khosravi.
\newblock Optimal binary space partitions for segments in the plane.
\newblock {\em Int. J. Comput. Geom. Appl.}, 22(3):187--206, 2012.
\newblock \href {https://doi.org/10.1142/S0218195912500045}
  {\path{doi:10.1142/S0218195912500045}}.

\bibitem{DoblerJJMMN-GD24}
Alexander Dobler, Michael J{\"{u}}nger, Paul~J. J{\"{u}}nger, Julian Meffert,
  Petra Mutzel, and Martin N{\"{o}}llenburg.
\newblock Revisiting {ILP} models for exact crossing minimization in storyline
  drawings.
\newblock In Stefan Felsner and Karsten Klein, editors, {\em 32nd International
  Symposium on Graph Drawing and Network Visualization (GD)}, volume 320 of
  {\em LIPIcs}, pages 31:1--31:19. Schloss Dagstuhl~-- Leibniz-Zentrum
  f{\"{u}}r Informatik, 2024.
\newblock \href {https://doi.org/10.4230/LIPICS.GD.2024.31}
  {\path{doi:10.4230/LIPICS.GD.2024.31}}.

\bibitem{dnsvw-cmtis-EuroCG23}
Alexander Dobler, Martin N\"ollenburg, Daniel Stojanovic, Ana{\"\i}s Villedieu,
  and Jules Wulms.
\newblock Crossing minimization in time interval storylines.
\newblock In Clemens Huemer and Carlos Seara, editors, {\em Proc. Europ.
  Workshop Comput. Geom. (EuroCG)}, pages 36:1--36:7, 2023.
\newblock URL: \url{https://arxiv.org/abs/2302.14213}.

\bibitem{fn-mwsv-GD17}
Theresa Fr\"oschl and Martin N\"ollenburg.
\newblock Minimzing wiggles in storyline visualizations.
\newblock In Fabrizio Frati and Kwan-Liu Ma, editors, {\em Graph Drawing \&
  Network Vis. (GD)}, volume 10692 of {\em LNCS}, pages 585--587. Springer,
  2018.
\newblock URL: \url{https://www.ac.tuwien.ac.at/files/pub/fn-mwsv-18.pdf}.

\bibitem{froschl_minimizing_2018}
Theresa Fröschl.
\newblock Minimizing wiggles in storyline visualizations.
\newblock Master's thesis, Technische Universität Wien, 2018.
\newblock \href {https://doi.org/10.34726/hss.2018.53581}
  {\path{doi:10.34726/hss.2018.53581}}.

\bibitem{fn-mwsv-18}
Theresa Fröschl and Martin Nöllenburg.
\newblock Minimzing wiggles in storyline visualizations.
\newblock In Fabrizio Frati and Kwan-Liu Ma, editors, {\em Graph Drawing and
  Network Visualization (GD'17)}, volume 10692 of {\em LNCS}, pages 585--587.
  Springer, 2018.
\newblock URL: \url{https://www.ac.tuwien.ac.at/files/pub/fn-mwsv-18.pdf}.

\bibitem{grimm_rolling-stock_2025}
Boris Grimm, Rowan Hoogervorst, and Ralf Borndörfer.
\newblock A comparison of two models for rolling stock scheduling.
\newblock {\em Transport. Sci.}, 2025.
\newblock Ahead of print.
\newblock \href {https://doi.org/10.1287/trsc.2024.0505}
  {\path{doi:10.1287/trsc.2024.0505}}.

\bibitem{gronemann2016}
Martin Gronemann, Michael J{\"u}nger, Frauke Liers, and Francesco Mambelli.
\newblock Crossing minimization in storyline visualization.
\newblock In Yifan Hu and Martin N{\"o}llenburg, editors, {\em Graph Drawing \&
  Network Vis. (GD)}, volume 8901 of {\em LNCS}, pages 367--381. Springer,
  2016.
\newblock \href {https://doi.org/10.1007/978-3-319-50106-2_29}
  {\path{doi:10.1007/978-3-319-50106-2_29}}.

\bibitem{hegemann-wolff-GD24}
Tim Hegemann and Alexander Wolff.
\newblock Storylines with a protagonist.
\newblock In Stefan Felsner and Karsten Klein, editors, {\em 32nd International
  Symposium on Graph Drawing and Network Visualization (GD)}, volume 320 of
  {\em LIPIcs}, pages 26:1--26:22. Schloss Dagstuhl~-- Leibniz-Zentrum
  f{\"{u}}r Informatik, 2024.
\newblock \href {https://doi.org/10.4230/LIPICS.GD.2024.26}
  {\path{doi:10.4230/LIPICS.GD.2024.26}}.

\bibitem{DBLP:journals/combinatorica/Karmarkar84}
Narendra Karmarkar.
\newblock A new polynomial-time algorithm for linear programming.
\newblock {\em Comb.}, 4(4):373--396, 1984.
\newblock \href {https://doi.org/10.1007/BF02579150}
  {\path{doi:10.1007/BF02579150}}.

\bibitem{kch-tgdt-AVI10}
{Nam Wook} Kim, Stuart~K. Card, and Jeffrey Heer.
\newblock Tracing genealogical data with timenets.
\newblock In Giuseppe Santucci, editor, {\em Advanced Visual Interfaces (AVI)},
  pages 241--248. ACM Press, 2010.
\newblock \href {https://doi.org/10.1145/1842993.1843035}
  {\path{doi:10.1145/1842993.1843035}}.

\bibitem{knuth93}
Donald~E. Knuth.
\newblock {\em The Stanford GraphBase: {A} platform for combinatorial
  computing}.
\newblock ACM Press, 1994.

\bibitem{knpss-mcsv-GD15}
Irina Kostitsyna, Martin N\"ollenburg, Valentin Polishchuk, Andr\'e Schulz, and
  Darren Strash.
\newblock On minimizing crossings in storyline visualizations.
\newblock In Emilio Di~Giacomo and Anna Lubiw, editors, {\em Graph Drawing \&
  Network Vis. (GD)}, volume 9411 of {\em LNCS}, pages 192--198. Springer,
  2015.
\newblock URL: \url{http://arxiv.org/abs/1509.00442}, \href
  {https://doi.org/10.1007/978-3-319-27261-0_16}
  {\path{doi:10.1007/978-3-319-27261-0_16}}.

\bibitem{liu_storyflow_2013}
Shixia Liu, Yingcai Wu, Enxun Wei, Mengchen Liu, and Yang Liu.
\newblock Storyflow: Tracking the evolution of stories.
\newblock {\em {IEEE} Trans. Vis. Comput. Graph.}, 19(12):2436--2445, 2013.
\newblock \href {https://doi.org/10.1109/TVCG.2013.196}
  {\path{doi:10.1109/TVCG.2013.196}}.

\bibitem{lwwll-stes-TVCG13}
Shixia Liu, Yingcai Wu, Enxun Wei, Mengchen Liu, and Yang Liu.
\newblock {StoryFlow}: Tracking the evolution of stories.
\newblock {\em IEEE Trans. Visual. Comput. Graphics}, 19(12):2436--2445, 2013.
\newblock \href {https://doi.org/10.1109/TVCG.2013.196}
  {\path{doi:10.1109/TVCG.2013.196}}.

\bibitem{strungemathiesenAestheticsOrderingStacked2021}
Steffen~Strunge Mathiesen and Hans{-}J{\"{o}}rg Schulz.
\newblock Aesthetics and ordering in stacked area charts.
\newblock In Amrita Basu, Gem Stapleton, Sven Linker, Catherine Legg, Emmanuel
  Manalo, and Petrucio Viana, editors, {\em Proc. Diagrammatic Representation
  and Inference (Diagrams'2021)}, volume 12909 of {\em Lecture Notes in
  Computer Science}, pages 3--19. Springer, 2021.
\newblock \href {https://doi.org/10.1007/978-3-030-86062-2\_1}
  {\path{doi:10.1007/978-3-030-86062-2\_1}}.

\bibitem{xkcd-storylines}
Randall Munroe.
\newblock Movie narrative charts.
\newblock Diagram available at \url{https://xkcd.com/657/}, 2009.
\newblock Accessed 2017/04/03.

\bibitem{ogawa_software_2010}
Michael Ogawa and Kwan{-}Liu Ma.
\newblock Software evolution storylines.
\newblock In Alexandru~C. Telea, Carsten G{\"{o}}rg, and Steven~P. Reiss,
  editors, {\em ACM Symposium on Software Visualization (SoftVis)}, pages
  35--42, 2010.
\newblock \href {https://doi.org/10.1145/1879211.1879219}
  {\path{doi:10.1145/1879211.1879219}}.

\bibitem{DBLP:journals/tsmc/SugiyamaTT81}
Kozo Sugiyama, Shojiro Tagawa, and Mitsuhiko Toda.
\newblock Methods for visual understanding of hierarchical system structures.
\newblock {\em {IEEE} Trans. Syst. Man Cybern.}, 11(2):109--125, 1981.
\newblock \href {https://doi.org/10.1109/TSMC.1981.4308636}
  {\path{doi:10.1109/TSMC.1981.4308636}}.

\bibitem{tanahashi_efficient_2015}
Yuzuru Tanahashi, Chien{-}Hsin Hsueh, and Kwan{-}Liu Ma.
\newblock An efficient framework for generating storyline visualizations from
  streaming data.
\newblock {\em {IEEE} Trans. Vis. Comput. Graph.}, 21(6):730--742, 2015.
\newblock \href {https://doi.org/10.1109/TVCG.2015.2392771}
  {\path{doi:10.1109/TVCG.2015.2392771}}.

\bibitem{tanahashi_design_2012}
Yuzuru Tanahashi and Kwan{-}Liu Ma.
\newblock Design considerations for optimizing storyline visualizations.
\newblock {\em {IEEE} Trans. Vis. Comput. Graph.}, 18(12):2679--2688, 2012.
\newblock \href {https://doi.org/10.1109/TVCG.2012.212}
  {\path{doi:10.1109/TVCG.2012.212}}.

\bibitem{vandijk17}
{Thomas}~C. {van Dijk}, {Martin} {Fink}, {Norbert} {Fischer}, {Fabian} {Lipp},
  {Peter} {Markfelder}, {Alexander} {Ravsky}, {Subhash} {Suri}, and {Alexander}
  {Wolff}.
\newblock Block crossings in storyline visualizations.
\newblock {\em J. Graph Alg. Appl.}, 21(5):873--913, 2017.
\newblock \href {https://doi.org/10.7155/jgaa.00443}
  {\path{doi:10.7155/jgaa.00443}}.

\bibitem{dlmw-csfbc-GD17}
Thomas~C. van Dijk, Fabian Lipp, Peter Markfelder, and Alexander Wolff.
\newblock Computing storylines with few block crossings.
\newblock In Fabrizio Frati and Kwan-Liu Ma, editors, {\em Graph Drawing \&
  Network Vis. (GD)}, volume 10692 of {\em LNCS}, pages 365--378. Springer,
  2018.
\newblock URL: \url{https://arxiv.org/abs/1709.01055}, \href
  {https://doi.org/10.1007/978-3-319-73915-1_29}
  {\path{doi:10.1007/978-3-319-73915-1_29}}.

\end{thebibliography}

\clearpage
\appendix

\section{Detailed Derivation of \texorpdfstring{\cref{thmt@@equationdx}}{}}\label{apx:geo}

We want to derive that given the radii $\rtc'$ and $\rtc''$ of the two arcs forming the character curve
for character~$c$ between time steps~$t$ and~$t+1$, the required horizontal space~\dxt{} fulfills
\equationdx*
Let $r_{t,c}$ be the sum of $\rtc'$ and $\rtc''$.  Let \oct{} be the distance
between $(x_{t,c}, y_{t,c})$ and $(x_{t+1,c}, y_{t+1,c})$.  See \cref{fig:geo-2} for illustration.
We see that there is a right-angled triangle with sides of length \dytc, \dxt{} and
\oct, respectively.  Furthermore, there is an isosceles triangle with base length
\oct{} and legs of length \rtc{} that has an altitude of length \dxt.
Therefore, $\oct^2 = \dxt^2 + \dytc^2$ and
$\dxt = \frac{\oct}{2 \rtc} \sqrt{4 \rtc^2 - \oct^2}$.

By substitution, we get
\begin{align*}
    \dxt^2                               & = \frac{\oct^2}{4\rtc^2}\left(4\rtc^2-\oct^2\right)          \\
    \frac{4\rtc^2\dxt^2}{\oct^2}         & = 4\rtc^2-\oct^2                                             \\
    \frac{4\rtc^2\dxt^2}{\dxt^2+\dytc^2} & = 4\rtc^2-\dxt^2-\dytc^2                                     \\
    4\rtc^2\dxt^2                        & = 4\rtc^2\dxt^2-\dxt^4-2\dxt^2\dytc^2-\dytc^4+4\rtc^2\dytc^2 \\
    \left(\dxt^2+\dytc^2\right)^2        & = 4\rtc^2\dytc^2                                             \\
    \dxt^2                               & = 2\rtc\dytc - \dytc^2.
\end{align*}

\begin{figure}[htb]
\centering
\includegraphics[page=6]{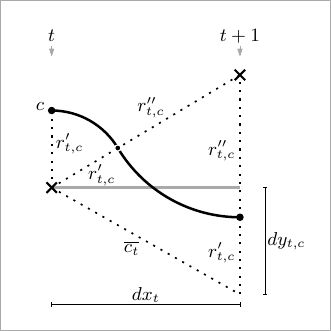}
\caption{The distance $dx$ of two consecutive time steps depends on $dy_i$ and $r_i$}
\label{fig:geo-2}
\end{figure}

\clearpage

\section{Full Page Versions of \texorpdfstring{\cref{fig:case-study-ddz}}{} and \texorpdfstring{\cite[Fig.~1]{grimm_rolling-stock_2025}}{}}
\label{apx:fullscreen}

\begin{sidewaysfigure}
    \includegraphics[width=\textwidth,trim={25mm 450mm 40mm 25mm}]{ddz.pdf}
    \caption{Storyline visualization of a rolling stock schedule using our drawing style (reprise of \cref{fig:case-study-ddz}).}
    \label{fig:case-study-ddz-large}
\end{sidewaysfigure}
\begin{sidewaysfigure}
    \includegraphics[width=\textwidth]{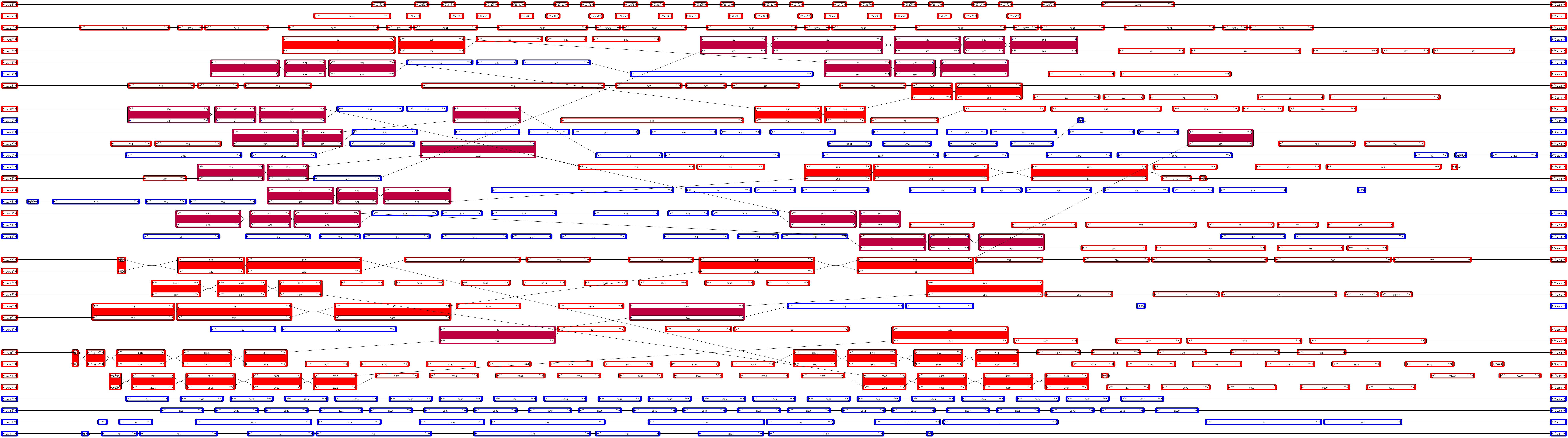}
    \caption{Visualization of a rolling stock schedule from \cite[Fig.~1]{grimm_rolling-stock_2025}.}
    \label{fig:grimm-ddz-large}
\end{sidewaysfigure}

\end{document}